\documentclass[11pt]{article}
\usepackage{geometry} 
\usepackage{graphicx}
\usepackage{amsmath,amssymb,amsthm,amsfonts,dsfont,mathtools}
\usepackage{enumitem}
\usepackage{relsize}
\usepackage[margin=1cm]{caption}
\usepackage{wrapfig}
\usepackage{frame,color}
\usepackage{environ,subfig}
\usepackage{hyperref}
\usepackage{xspace}
\usepackage{framed}
\usepackage{comment}
\usepackage{changepage}
\usepackage[linesnumbered,boxed]{algorithm2e}
\usepackage{algorithmic}
\usepackage{mathrsfs}
\usepackage{tabu} 
\usepackage{authblk}

\usepackage{thmtools}
\usepackage{thm-restate}

\newcommand{\decomposition}{expander decomposition}
\newcommand{\lowdiamdecomp}{$\mathtt{LowDiamDecomposition}$}

\newcommand{\mpx}{$\mathtt{Clustering}$}
\newcommand{\nibble}{$\mathtt{Nibble}$}
\newcommand{\distnibble}{$\mathtt{Approximate Nibble}$}
\newcommand{\randnibble}{$\mathtt{Random Nibble}$}
\newcommand{\parallelnibble}{$\mathtt{Parallel Nibble}$}
\newcommand{\balancedcut}{$\mathtt{Partition}$}
\newcommand{\tlast}{t_{0}}
\newcommand{\bal}{\operatorname{bal}}

\newcommand{\fff}{f}

\newcommand{\ppp}{\tilde{p}}
\newcommand{\rhoo}{\tilde{\rho}}

\newcommand{\jmax}{j_{\operatorname{max}}}

\newcommand{\vol}{\operatorname{Vol}}

\def\tpi{\tilde{\pi}_t}

\newtheorem{lemma}{Lemma}

\newtheorem{definition}{Definition}

\newcommand{\ignore}[1]{}

\newcommand{\Expect}{\mathbf{E}}

\newcommand{\Prob}{\mathbf{Pr}}
\newcommand{\bydef}{\stackrel{\operatorname{def}}{=}}

\newcommand{\ceil}[1]{\left\lceil #1 \right\rceil}
\newcommand{\floor}[1]{\lfloor #1 \rfloor}

\newcommand{\poly}{{\operatorname{poly}}}

\newcommand{\dist}{\operatorname{dist}}

\newcommand{\ID}{\operatorname{ID}}
\newcommand{\LOCAL}{\mathsf{LOCAL}}

\newcommand{\CONGEST}{\mathsf{CONGEST}}
\newcommand{\CLIQUE}{\mathsf{CONGESTED}\text{-}\mathsf{CLIQUE}}

\newcommand{\mix}{\ensuremath{\tau_{\operatorname{mix}}}}

\newcommand{\Erdos}{Erd\H{o}s}
\newcommand{\Renyi}{R\'{e}nyi}

\def\Remove#1{\textsf{Remove}\text{-}#1}

\title{
Improved Distributed Expander Decomposition  and Nearly Optimal Triangle Enumeration\thanks{This work is supported by NSF grants CCF-1514383, CCF-1637546, and CCF-1815316.}}

\author[1]{Yi-Jun Chang}
\author[2]{Thatchaphol Saranurak}

\affil[1]{University of Michigan, USA}
\affil[2]{Toyota Technological Institute at Chicago, USA}

\begin{document}
\date{}
\maketitle
\setcounter{page}{1}


\begin{center}
  {\bf Abstract} 
\end{center}

An $(\epsilon,\phi)$-expander decomposition of a graph $G=(V,E)$
is a clustering of the vertices $V=V_{1}\cup\cdots\cup V_{x}$ such
that (1) each cluster $V_{i}$ induces subgraph with conductance at
least $\phi$, and (2) the number of inter-cluster edges is at most
$\epsilon|E|$. 
%
In this paper, we give an improved distributed expander decomposition,
and obtain a nearly optimal distributed triangle enumeration algorithm
in the $\CONGEST$ model. 

Specifically, we construct an $(\epsilon,\phi)$-expander
decomposition with $\phi=(\epsilon/\log n)^{2^{O(k)}}$ in $O(n^{2/k}\cdot\text{poly}(1/\phi,\log n))$
rounds for any $\epsilon\in(0,1)$ and positive integer $k$. For
example, a $(1/n^{o(1)},1/n^{o(1)})$-expander decomposition only requires $O(n^{o(1)})$ rounds to compute, which is optimal up to subpolynomial factors,
and a $(0.01,1/\text{poly}\log n)$-expander decomposition can be computed in $O(n^{\gamma})$ rounds, for any arbitrarily small constant $\gamma>0$.
Previously, the algorithm by Chang, Pettie, and Zhang can construct
a $(1/6,1/\text{poly}\log n)$-expander decomposition using $\tilde{O}(n^{1-\delta})$
rounds for any $\delta>0$, with a caveat that the algorithm is allowed to throw
away a set of edges into an extra part which form a subgraph with arboricity
at most $n^{\delta}$. Our algorithm does not have this caveat.

By slightly modifying the distributed algorithm for routing on expanders
by Ghaffari, Kuhn and Su {[}PODC'17{]}, we  obtain a triangle
enumeration algorithm using $\tilde{O}(n^{1/3})$ rounds. This matches
the lower bound by Izumi and Le~Gall {[}PODC'17{]} and Pandurangan,
Robinson and Scquizzato {[}SPAA'18{]} of $\tilde{\Omega}(n^{1/3})$
which holds even in the $\CLIQUE$ model.
To the best of our knowledge, this provides the first non-trivial example for a distributed problem that has essentially the same complexity (up to a polylogarithmic factor) in both $\CONGEST$ and $\CLIQUE$.

The key technique in our proof is the first distributed approximation
algorithm for finding a low conductance cut that is as balanced as
possible. Previous distributed sparse cut algorithms do not have this nearly most balanced guarantee.\footnote{Kuhn and Molla~\cite{KuhnM15} previously claimed that their approximate
sparse cut algorithm also has the nearly most balanced guarantee, but this claim turns
out to be incorrect~\cite[Footnote 3]{ChangPZ19}. \label{footnote:wrong}} 

\newpage
\section{Introduction}  \label{XXX-sect-intro}

In this paper, we consider the task of finding an \emph{expander decomposition} of a distributed network in the $\CONGEST$ model of distributed computing. Roughly speaking, an expander decomposition of a graph $G = (V,E)$ is a clustering of the vertices $V = V_1 \cup \cdots \cup V_x$ such that (1) each component $V_i$ induces a high conductance subgraph, and (2) the number of inter-component edges is small.
This natural bicriteria optimization problem of finding a good expander decomposition was introduced by Kannan Vempala and Vetta~\cite{KannanVV04}, and was further studied in many other subsequent works~\cite{SpielmanT13,OrecchiaA14,PatrascuT07,SanjeevSU09,Trevisan08,OrecchiaV11,SaranurakW19}.\footnote{
The existence of the expander decomposition is (implicitly) exploited first in the context of property testing \cite{GoldreichR1999}.}
 The expander decomposition has a wide range of applications, and it has been applied to solving linear systems~\cite{SpielmanT14}, unique games~\cite{AroraBS2015,Trevisan08,RaghavendraS2010},
 minimum cut~\cite{KawarabayashiT15}, and dynamic algorithms~\cite{NanongkaiSW17}. 
 
 Recently, Chang, Pettie, and Zhang~\cite{ChangPZ19} applied this technique to the field of distributed computing, and they showed that a variant of expander decomposition can be computed efficiently in $\CONGEST$. 
 Using this decomposition, they showed that triangle detection and enumeration can be solved in
$\tilde{O}(n^{1/2})$ rounds.\footnote{The $\tilde{O}(\cdot)$ notation hides any polylogarithmic factor.}
The previous state-of-the-art 
bounds for  triangle detection and enumeration  were 
$\tilde{O}(n^{2/3})$ and $\tilde{O}(n^{3/4})$, respectively,
due to Izumi and Le Gall~\cite{IzumiL17}.
Later, Daga et~al.~\cite{Daga2019distributed} exploit this decomposition and obtain the first algorithm for computing edge connectivity of a graph \emph{exactly} using sub-linear number of rounds. 

 Specifically, the variant of the decomposition in \cite{ChangPZ19} is as follows. If we allow one extra part that induces an $n^{\delta}$-arboricity subgraph\footnote{The arboricity of a graph is the minimum number $\alpha$ such that its edge set can be partitioned into $\alpha$ forests.} in the decomposition, then in $O(n^{1 - \delta})$ rounds we can construct an expander decomposition in $\CONGEST$ such that each component has $1/O(\poly \log n)$ conductance and the number of inter-component edges is at most $|E|/6$.

A major open problem left by the work~\cite{ChangPZ19} is to design an efficient distributed algorithm constructing an expander decomposition without the extra low-arboricity part. In this work, we show that this is possible.
A consequence of our new expander decomposition algorithm is that triangle enumeration can be solved in $O(n^{1/3} \poly \log n)$ rounds, nearly matching the $\Omega(n^{1/3}/\log n)$ lower bound~\cite{IzumiL17,PanduranganRS18} by a polylogarithmic factor.

\paragraph{The $\CONGEST$ Model.}
In the $\CONGEST$ model of distributed computing, the underlying distributed network
is represented as an undirected graph $G=(V,E)$, where each vertex corresponds to a computational device, and each edge corresponds to 
a bi-directional communication link.
Each vertex $v$ has a distinct $\Theta(\log n)$-bit identifier $\ID(v)$.
The computation proceeds according to synchronized \emph{rounds}.
In each round, each vertex $v$ can perform unlimited 
local computation, and may send a distinct
$O(\log n)$-bit message to each of its neighbors.
Throughout the paper we only consider the randomized variant of $\CONGEST$.
Each vertex is allowed to generate unlimited local random bits, 
but there is no global randomness. We say that an algorithm succeeds {\em with high probability} (w.h.p.) if its failure probability is at most $1 / \poly(n)$.

The $\CLIQUE$ model is a variant of 
$\CONGEST$ that allows all-to-all communication, and the $\LOCAL$ model is a variant of 
$\CONGEST$ that allows messages of unbounded length.

\paragraph{Terminology.} Before we proceed, we review the graph terminologies related to the expander decomposition.  Consider a graph $G = (V,E)$. For a vertex subset $S$, we write $\vol(S)$ to denote $\sum_{v \in S} \deg(v)$. Note that by default the degree is with respect to the original graph $G$. We write $\bar{S} = V \setminus S$, and let $\partial(S) = E(S, \bar{S})$ be the set of edges $e = \{u,v\}$ with $u \in S$ and $v \in \bar{S}$.  The {\em sparsity} or \emph{conductance} of a cut $(S, \bar{S})$ is defined as $\Phi(S) = |\partial(S)| /  \min\{\vol(S), \vol(\bar{S})\}$.
The {\em conductance} $\Phi_G$ of a graph $G$ 
is the minimum value of $\Phi(S)$ over all vertex subsets $S$.
Define the {\em balance} $\bal(S)$ of a cut $S$ by $\bal(S) = \min\{\vol(S), \vol(\bar{S})\} / \vol(V)$. We say that $S$ is a \emph{most-balanced} cut of $G$ of conductance at most $\phi$ if $\bal(S)$ is maximized among all cuts of $G$ with conductance at most $\phi$.
We have the following relation~\cite{JerrumS89} between the 
mixing time $\mix(G)$ and conductance $\Phi_G$: 
\[
\Theta\left(\frac{1}{\Phi_G}\right) \leq \mix(G) \leq \Theta\left(\frac{\log n}{\Phi_G^2}\right).
\]

Let $S$ be a vertex set. Denote $E(S)$ by the set of all edges whose two endpoints are both within $S$. We write $G[S]$ to denote the subgraph induced by $S$, and we write $G\{S\}$ to denote the graph resulting from adding $\deg_V(v) - \deg_S(v)$ self loops to each vertex $v$ in  $G[S]$. Note that the degree of each vertex $v \in S$ in both $G$ and $G\{S\}$ is identical. As in~\cite{SpielmanS08}, each self loop of $v$ contributes 1 in the calculation of $\deg(v)$. Observe that we always have
\[\Phi(G\{S\}) \leq \Phi(G[S]).\]

Let $v$ be a vertex. Denote $N(v)$ as the set of neighbors of $v$. We also write $N^k(v) = \{ u \in V \ | \ \dist(u,v) \leq k\}$. Note that $N^1(v) = N(v) \cup \{v\}$. These notations $\dist(u,v)$, $N(v)$, and $N^k(v)$ depend on the underlying graph $G$. When the choice of underlying graph is not clear from the context, we use a subscript to indicate the underlying graph we refer to.

\paragraph{Expander Decomposition.}
An $(\epsilon, \phi)$-\decomposition\ of a graph $G = (V,E)$ is defined as a partition of the vertex set $V = V_1 \cup \cdots \cup V_x$ satisfying the following conditions.
\begin{itemize}
    \item For each component $V_i$, we have $\Phi(G\{V_i\}) \geq \phi$.
    \item The number of inter-component edges $\left(|\partial(V_1)| + \cdots + |\partial(V_x)|\right)/2$  is at most $\epsilon |E|$.
\end{itemize}
The main contribution of this paper is the following result.

\begin{restatable}{theorem}{restateexpanderdecomposition}\label{XXX-thm-expander-decomposition}
Let $\epsilon \in (0,1)$, and let $k$ be a positive integer.
An $(\epsilon, \phi)$-\decomposition\ with $\phi = (\epsilon / \log n)^{2^{O(k)}}$ can be constructed in $O\left(n^{2/k} \cdot \poly(1/\phi, \log n)\right) = O\left(n^{2/k} \cdot \left(\frac{\log n}{\epsilon}\right)^{2^{O(k)}} \right)$ rounds, w.h.p.
\end{restatable}

The proof of Theorem~\ref{XXX-thm-expander-decomposition} is in Section~\ref{XXX-sect-expander-trim}.
We emphasize that the number of rounds does not depend on the diameter of $G$.
There is a trade-off between the two parameters $\epsilon$ and  $\phi$.
For example, an $(\epsilon, \phi)$-\decomposition\ with $\epsilon = 2^{-\log^{1/3}n}$ and  $\phi = 2^{-\log^{2/3}n}$ can be constructed in $n^{O(1/\log \log n)}$ rounds by setting $k = O(\log \log n)$ in Theorem~\ref{XXX-thm-expander-decomposition}. If  we are allowed to have $\epsilon = 0.01$ and spend $O(n^{0.01})$ rounds, then we can achieve $\phi = 1/ O(\poly \log n)$.

\paragraph{Distributed Triangle Finding.} 
Variants of the triangle
 finding problem have been studied in the literature~\cite{AbboudCK2017,Censor2016,ChangPZ19,DolevLP12,DruckerKO14,FischerGO18,PanduranganRS18,IzumiL17}. In the {\em triangle detection} problem, it is required that at least one vertex  must report a triangle if the graph has at least one triangle.  In the {\em triangle enumeration} problem, it is required that each triangle of the graph is reported by at least one vertex. Both of these problems can be solved in $O(1)$ rounds in $\LOCAL$. It is the bandwidth constraint of $\CONGEST$ and $\CLIQUE$ that makes these problems non-trivial.
 
 It is important that a triangle $T=\{u,v,w\}$  is allowed to be reported by a vertex $x \notin T$. If it is required that a triangle $T=\{u,v,w\}$ has to be reported by a vertex $x \in T$, then there is an $\Omega(n / \log n)$ lower bound~\cite{IzumiL17} for triangle enumeration, in both $\CONGEST$ and $\CLIQUE$. To achieve a round complexity of $o(n / \log n)$, it is necessary that some triangles $T$ are reported by vertices not in $T$.

 Dolev, Lenzen, and Peled~\cite{DolevLP12} showed that triangle
 enumeration can be solved deterministically 
 in $O(n^{1/3}/\log n)$ rounds in $\CLIQUE$. This algorithm is optimal, as it matches the  $\Omega(n^{1/3}/\log n)$-round lower bound~\cite{IzumiL17,PanduranganRS18} in $\CLIQUE$.
 Interestingly, if we only want to detect one triangle or count the number of triangles, then  Censor-Hillel et al.~\cite{Censor2016} showed that the round complexity in $\CLIQUE$ can be improved to 
 $\tilde{O}(n^{1-(2/\omega) + o(1)}) = o(n^{0.158})$~time~\cite{Censor2016}, where $\omega < 2.373$ is the exponent of the complexity of matrix multiplication~\cite{LeGall14}.

For the $\CONGEST$ model,
Izumi and Le Gall~\cite{IzumiL17} showed that
the triangle detection and enumeration problems can be solved in $\tilde{O}(n^{2/3})$ and $\tilde{O}(n^{3/4})$ time, respectively. 
These upper bounds were later improved to $\tilde{O}(n^{1/2})$ by Chang, Pettie, and Zhang using a variant of expander decomposition~\cite{ChangPZ19}.

A consequence of Theorem~\ref{XXX-thm-expander-decomposition} is that triangle enumeration (and hence detection) can be solved in $\tilde{O}(n^{1/3})$ rounds, almost matching the $\Omega(n^{1/3} / \log n)$ lower bound~\cite{IzumiL17,PanduranganRS18} which holds even in $\CLIQUE$. To the best of our knowledge, this provides the first non-trivial example for a distributed problem that has essentially the same complexity (up to a polylogarithmic factor) in both $\CONGEST$ and $\CLIQUE$, i.e., allowing non-local communication links does not help. In contrast,  many other graph problems  can be solved much more efficiently in $\CLIQUE$ than in $\CONGEST$; see e.g.,~\cite{JurdzinskiN2018,GhaffariN2018}. 

\begin{restatable}{theorem}{restatetriangle}\label{XXX-thm-triangle-enum}
Triangle enumeration can be solved in  $\tilde{O}(n^{1/3})$ rounds in $\CONGEST$, w.h.p.
\end{restatable}

The proof of Theorem~\ref{XXX-thm-triangle-enum} is in Section~\ref{XXX-triangle-enum}. Note that Theorem~\ref{XXX-thm-triangle-enum} immediately implies an algorithm for triangle detection with the same number of rounds. However, while the best known lower bounds~\cite{AbboudCK2017,FischerGO18} for triangle detection can currently exclude only 1-round algorithms. Whether the large gap between upper and lower bounds for this problem can be closed remains an intriguing question.

\subsection{Prior Work on Expander Decomposition}\label{sec:related}


In the centralized setting, the first polynomial time algorithm for construction an $(\epsilon,\phi)$-expander decomposition is by Kannan, Vempala and Vetta~\cite{KannanVV04} where $\epsilon = \tilde{O}(\phi)$.
Afterward, Spielman and Teng~\cite{SpielmanT11,SpielmanT13} significantly improved the running time to be near-linear in $m$, where $m$ is the number of edges. In time $\tilde{O}(m/\poly(\phi))$, they can construct a ``weak'' $(\poly(\phi,\log n),\phi)$-expander decomposition. Their weak expander only has the following weaker guarantee that 
each part $V_i$ in the partition of $V$ might not induce an expander, and we only know that $V_i$ is contained in some {\em unknown} expander. That is, there exists some $W_i \supseteq V_i$ where $\Phi_{G\{W_i\}} \ge \phi$. Although this guarantee suffices for many applications (e.g.~\cite{KelnerLOS14,CohenKPPRSV17}), some other applications~\cite{NanongkaiSW17,ChuGPSSW18}, including the triangle  enumeration algorithm of~\cite{ChangPZ19}, crucially needs the fact that each part in the decomposition induces an expander.



Nanongkai and Saranurak~\cite{NanongkaiS17} and, independently, Wulff-Nilsen~\cite{Wulff-Nilsen17} gave a fast algorithm without weakening the guarantee as the one in~\cite{SpielmanT11,SpielmanT13}. In~\cite{NanongkaiS17}, their algorithm finds a $(\phi \log^{O(k)}n,\phi$)-expander decomposition in time $\tilde{O}(m^{1+1/k})$. Although the trade-off is worse in~\cite{Wulff-Nilsen17}, their high-level approaches are in fact  the same. They gave the same black-box reduction from constructing an expander decomposition to finding a nearly most balanced sparse cut. The difference only comes from the quality of their nearly most balanced sparse cuts algorithms. 
Our distributed algorithm will also follow this high-level approach.  

Most recently, Saranurak and Wang~\cite{SaranurakW19} gave a $(\tilde{O}(\phi),\phi)$-expander decomposition algorithm with running time $\tilde{O}(m/\phi)$. This is optimal up to a polylogarithmic factor when $\phi \ge 1/\poly\log (n)$.
We do not use their approach, as their {\em trimming step} seems to be inherently sequential and very challenging to parallelize or make distributed.

The only previous  expander decomposition in the distributed setting is by Chang, Pettie, and Zhang~\cite{ChangPZ19}. Their distributed algorithm gave an $(1/6,1/ \poly\log(n))$-expander decomposition with an extra part which is an $n^\delta$-arboricity subgraph in $O(n^{1-\delta})$ rounds in $\CONGEST$. Our distributed algorithm significantly improved upon this work.




\subsection{Technical Overview}\label{sec:tech_overview}

For convenience, we call a cut with conductance at most $\phi$
a \emph{$\phi$-sparse cut} in this section. To give a high-level
idea, the most straightforward algorithm for constructing an expander
decomposition of a graph $G=(V,E)$ is as follows. Find a $\phi$-sparse
cut $S$. If such a cut $S$ does not exist, then return $V$ as a part in the
partition. Otherwise, recurse on both sides $G\{S\}$ and $G\{V-S\}$,
and so the edges in $E(S,V-S)$ become inter-cluster edges.
To see the correctness,  once the recursion stops at $G\{U\}$ for some $U$, we know that $\Phi_{G\{U\}}\ge\phi$.
Also, the total number of inter-cluster edges is at most $O(m \phi\log n)$
because (1) each inter-cluster edge can be charged to edges in the
smaller side of some $\phi$-sparse cut, and (2) each edge can be
in the smaller side of the cut for at most $O(\log n)$ times. 

This straightforward approach has two efficiency issues: (1) checking
whether a $\phi$-sparse cut exists does not admit fast distributed
algorithms (and is in fact NP-hard), and (2) a $\phi$-sparse cut
$S$ can be very unbalanced and hence the recursion depth can be as
large as $\Omega(n)$. Thus, even if we ignore time spent on finding cuts,
the round complexity due to the recursion depth is too high. At a
high-level, all previous algorithms (both centralized and distributed)
handle the two issues in the same way up to some extent. First, they
instead use \emph{approximate sparse cut algorithms} which either
find some $\phi'$-sparse cut or certify that there is no $\phi$-sparse
cut where $\phi'\gg\phi$. Second, they find a cut with some guarantee about the balance of the cut, i.e., the smaller side of the cut should be sufficiently large. 

Let us contrast our approach with the only previous distributed expander
decomposition algorithm by Chang, Pettie, and Zhang~\cite{ChangPZ19}. They
gave an approximate sparse cut algorithm such that the smaller side
of the cut has $\Omega(n^{\delta})$ vertices for some constant $\delta>0$,
so the recursion depth is $O(n^{1-\delta})$. They  guarantee this
 property by ``forcing'' the graph to have minimum degree
at least $n^{\delta}$, so any $\phi$-sparse cut  must contain
$\Omega(n^{\delta})$ vertices (this uses the fact that the graph is simple)
To force the graph
to have high degree, they keep removing vertices with degree at most
$n^{\delta}$ at any step of the algorithms. Throughout the whole
algorithm, the removed part form a graph with arboricity at most $n^{\delta}$.
This explains why their decomposition outputs the extra part which induces a
low arboricity subgraph. With some other ideas on distributed implementation,
they obtained the round complexity of $\tilde{O}(n^{1-\delta})$, roughly matching the
recursion depth.

In this paper, we avoid this extra low-arboricity part. The key
component is the following. Instead of just guaranteeing that the
smaller side of the cut has $\Omega(n^{\delta})$ vertices, we give the
first
efficient distributed algorithm for computing a \emph{nearly most
balanced} sparse cut. Suppose  there is a $\phi$-sparse cut with
balance $b$, then our sparse cut algorithm returns a $\phi'$-sparse cut with
balance at least $\Omega(b)$, where $\phi'$ is not much larger than
$\phi$. 
Intuitively, given that we can find a nearly most balanced sparse
cut efficiently, the recursion depth should be made very small. This intuition can
be made formal using the ideas in the centralized setting from Nanongkai
and Saranurak~\cite{NanongkaiS17} and  Wullf-Nilsen~\cite{Wulff-Nilsen17}.
Our main technical contribution is two-fold. First, we show the
first distributed algorithm for computing a nearly most balanced sparse
cut, which is our key algorithmic tool. Second, in order to obtain
a fast distributed algorithm, we must modify the centralized approach
of~\cite{NanongkaiS17,Wulff-Nilsen17} on how to construct an expander decomposition. 
In particular, 
we need to run a {\em low diameter decomposition} whenever we encounter a graph with high diameter, as our distributed algorithm for finding a nearly most balanced sparse cut is fast only on graphs with low diameter.

\paragraph{Sparse Cut Computation.}
At a high level, our distributed nearly  most balanced sparse cut algorithm is a distributed implementation of the sequential algorithm of Spielman and Teng~\cite{SpielmanT13}.  The algorithm of~\cite{SpielmanT13} involves $\tilde{O}(m)$ \emph{sequential} iterations of \nibble\ with a random starting vertex  on the remaining subgraph. 
Roughly speaking, the procedure  \nibble\ aims at finding a sparse cut by simulating a random walk. The idea is that if the starting vertex $v$ belongs to some sparse cut $S$, then it is likely that most of the probability mass will be trapped inside $S$.
Chang, Pettie, and Zhang~\cite{ChangPZ19} showed that $\tilde{O}(m)$ simultaneous iterations of an approximate version of \nibble\ with a random starting vertex can be implemented efficiently in $\CONGEST$ in $O(\poly(1/\phi, \log n))$ rounds, where $\phi$ is the target conductance.
A major difference between this work and~\cite{ChangPZ19} is that the expander decomposition algorithm of~\cite{ChangPZ19} does not need any  requirement about the balance of the cut in their sparse cut computation. 

Note that the  $\tilde{O}(m)$ \emph{sequential} iterations of \nibble\ in the nearly most balanced sparse cut algorithm of~\cite{SpielmanT13} cannot be completely parallelized. For example, it is possible that the union of all   $\tilde{O}(m)$ output of  \nibble\ equals the entire graph. Nonetheless, we show that this process can be partially parallelized at the cost of worsening the conductance guarantee by a polylogarithmic factor.

\begin{restatable}[Nearly most balanced sparse cut]{theorem}{restatenearlybalcut}\label{XXX-thm-nearly-bal-cut}
Given a parameter $\phi = O(1/ \log^5 n)$, there is an $O(D \cdot \poly(\log n, 1/\phi))$-round algorithm $\mathcal{A}$ that achieves the following w.h.p.
\begin{itemize}
    \item In case $\Phi(G) \leq \phi$, the algorithm $\mathcal{A}$ is guaranteed to return a cut $C$ with balance  $\bal(C) \geq \min\{b/2, 1/48\}$ and conductance $\Phi(C) = O(\phi^{1/3} \log^{5/3} n)$, where $b$ is defined as $b =\bal(S)$, where $S$ is a most-balanced sparse cut of $G$ of conductance at most $\phi$.
    \item In case $\Phi(G) > \phi$, the algorithm $\mathcal{A}$ either returns $C = \emptyset$ or returns a cut $C$ with conductance $\Phi(C) = O(\phi^{1/3} \log^{5/3} n)$.
\end{itemize}
\end{restatable}

The proof of Theorem~\ref{XXX-thm-nearly-bal-cut} is in Appendix~\ref{XXX-sect-balanced-cut}.
We note again that this is the first distributed sparse cut algorithm with a nearly most balanced guarantee. 
The problem of finding a sparse cut the distributed setting has been studied prior to the work of~\cite{ChangPZ19}.
Given that there is a $\phi$-sparse cut and balance $b$, the algorithm of Das Sarma, Molla, and Pandurangan~\cite{DasSarma15} finds a cut of conductance at most $\tilde{O}(\sqrt{\phi})$ in $\tilde{O}((n + (1/\phi))/b)$ rounds in $\CONGEST$. The round complexity was later improved to $\tilde{O}(D + 1/(b \phi))$ by Kuhn and Molla~\cite{KuhnM15}. These prior works have the following drawbacks: (1) their running time depends on $b$ which can be as small as $O(1/n)$, and (2) their output cuts are not guaranteed to be nearly most balanced (see footnote \ref{footnote:wrong}).

\paragraph{Low Diameter Decomposition.}  
The runtime of our distributed sparse cut algorithm (Theorem~\ref{XXX-thm-nearly-bal-cut}) is proportional to the diameter.
To avoid running this algorithm on a high diameter graph, we employ a 
low diameter decomposition to  decompose the current graph into components of small diameter.

The low diameter decomposition algorithm of Miller, Peng, and Xu~\cite{miller2013parallel} can already be implemented in $\CONGEST$ efficiently. Roughly, their algorithm is to let each vertex $v$ sample
$\delta_v\sim \text{Exponential}(\beta)$, $\beta \in (0,1)$, and then  $v$ is assigned to the cluster of $u$ that minimizes $\dist(u,v) - \delta_u$. A similar approach has been applied to construct a \emph{network decomposition}~\cite{doi:10.1002/rsa.3240050305,Linial1993}.

However, there is one subtle issue that the guarantee that the number of inter-cluster edges is at most $O(\beta |E|)$ only holds \emph{in expectation}. In sequential or parallel computation model, we can simply repeat the procedure for several 
times and take the best result. 
In $\CONGEST$, this however takes at least diameter time, which is inefficient when the diameter is large.


We provide a technique that allows us to achieve this guarantee {\em with high probability} without spending diameter time, so we can  ensure that the number of inter-cluster edges is small with high probability in our expander decomposition algorithm.\footnote{We remark that the  triangle enumeration algorithm of~\cite{ChangPZ19} still works even if the guarantee on the number of inter-cluster edges in the expander decomposition only holds in expectation.}

Intuitively, the main barrier needed to be overcome is the high dependence among the $|E|$ events that an edge $\{u,v\}$ has its endpoints in different clusters.
Our strategy is to compute a  partition $V = V_D \cup V_S$ in such a way that $V_D$ already induces a low diameter clustering, and the edges incident to $V_S$ satisfy the property that if we run the the low diameter decomposition algorithm of~\cite{miller2013parallel}, the events that they are inter-cluster have sufficiently small dependence. Then we can use a variant of Chernoff bound with bounded dependence \cite{Pemmaraju01} to bound the number of inter-cluster edges with high probability.

\begin{restatable}[Low diameter decomposition]{theorem}{restatelowdiamclustering}
\label{XXX-thm-low-diam-clustering}
Let  $\beta \in (0,1)$. There is an $O\left(\poly(\log n, 1/\beta)\right)$-round  algorithm $\mathcal{A}$ that finds
a partition of the vertex set $V = V_1 \cup \cdots \cup V_x$ satisfying the following conditions w.h.p.
\begin{itemize}
    \item Each component $V_i$ has diameter  $O\left(\frac{\log^2 n}{\beta^2}\right)$.
    \item The number of inter-component edges $\left(|\partial(V_1)| + \cdots + |\partial(V_x)|\right)/2$  is at most $\beta |E|$.
\end{itemize}
\end{restatable}

Adapting the algorithm of~\cite{miller2013parallel} to $\CONGEST$, in   $O\left(\frac{\log n}{\beta}\right)$ rounds we can decompose the graph into components of diameter $O\left(\frac{\log n}{\beta}\right)$ such that the number of inter-component edges is $O(\beta |E|)$ \emph{in expectation}. In Appendix~\ref{XXX-sect-low-diam-clustering} we extend this result to obtain a high probability bound and prove Theorem~\ref{XXX-thm-low-diam-clustering}.






\paragraph{Triangle Enumeration.}
Incorporating our expander decomposition algorithm (Theorem~\ref{XXX-thm-expander-decomposition}) with the triangle enumeration algorithm of~\cite{ChangPZ19,GhaffariKS17}, we immediately obtain an $\tilde{O}(n^{1/3}) \cdot 2^{O(\sqrt{\log n})}$-round algorithm for triangle enumeration.
This round complexity can be further improved to $\tilde{O}(n^{1/3})$ by adjusting the routing algorithm of Ghaffari, Kuhn, and Su~\cite{GhaffariKS17} on graphs of small mixing time. 
The main observation is their algorithm can be viewed as a distributed data structure with a trade-off between the query time and the pre-processing time. In particular, for any given constant $\epsilon > 0$, it is possible to achieve $O(\poly \log n)$ query time by spending $O(n^{\epsilon})$ time on pre-processing.

\section{Expander Decomposition} \label{XXX-sect-expander-trim}
The goal of this section is to prove Theorem~\ref{XXX-thm-expander-decomposition}.

\restateexpanderdecomposition*

For the sake of convenience, we denote
\[h(\theta) = \Theta\left(\theta^{1/3} \log^{5/3} n\right)\]
as an increasing function associated with Theorem~\ref{XXX-thm-nearly-bal-cut} such that when we run the nearly most balanced sparse cut algorithm of Theorem~\ref{XXX-thm-nearly-bal-cut} with conductance parameter $\theta$, if the output subset $C$ is non-empty, then it has $\Phi(C) \leq h(\theta)$. We note that  
\[h^{-1}(\theta) = \Theta\left({\theta^3 / \log^5 n}\right).\]
Let $\epsilon \in (0,1)$ and $k \geq 1$ be the parameters specified in Theorem~\ref{XXX-thm-expander-decomposition}.
We define the following parameters that are used in our algorithm.

\begin{description}
\item[Nearly Most Balanced Sparse Cut:] We define $\phi_0 = O(\epsilon^2 / \log^7 n)$ in such a way that when we run the nearly most balanced sparse cut algorithm with this conductance parameter, any non-empty output $C$ must satisfy  $\Phi(C) \leq h(\phi_0) = \frac{\epsilon/6} {\log {n \choose 2}}$. For each $1 \leq i \leq k$, we define $\phi_i = h^{-1}(\phi_{i-1})$.
\item[Low Diameter Decomposition:] The parameter $\beta = O(\epsilon^2 / \log n)$ for the low diameter decomposition is chosen as follows. Set $d  = O( (1/\epsilon) \log n)$ as the smallest integer such that $(1 - \epsilon/12)^{d} \cdot  2{n \choose 2} < 1$. Then we define $\beta =  (\epsilon/3) / d$.
\end{description}

We  show that an $(\epsilon, \phi)$-\decomposition\ can be constructed in $O\left(n^{2/k} \cdot \poly(1/\phi, \log n)\right)$ rounds, with conductance parameter $\phi = \phi_k =  (\epsilon / \log n)^{2^{O(k)}}$. We will later see that $\phi = \phi_k$ is the smallest conductance parameter we ever use for applying the nearly most balanced sparse cut algorithm.

\paragraph{Algorithm.} Our algorithm has two phases. In the algorithm there are three places where we remove edges from the graph, and they are tagged with \Remove{$j$}, for $1 \leq j \leq 3$ for convenience. Whenever we remove an edge $e = \{u, v\}$, we add a self loop at both $u$ and $v$, and so the degree of a vertex never changes throughout the algorithm. We never remove self loops.

At the end of the algorithm, $V$ is partitioned into connected components $V_1, \ldots, V_x$ induced by the remaining edges. To prove the correctness of the algorithm, we will show that the number of removed edges is at most $\epsilon|E|$, and $\Phi_{G\{V_i\}} \geq \phi$ for each component $V_i$.

\begin{framed}
\noindent {\bf Phase 1.} 

\medskip
The input graph is $G = (V,E)$.
\begin{enumerate}
    \item 
    Do the low diameter decomposition algorithm (Theorem~\ref{XXX-thm-low-diam-clustering}) with parameter $\beta$ on $G$. Remove all inter-cluster edges (\Remove{1}).
    \item For each connected component $U$ of the graph, do the nearly most balanced sparse cut algorithm (Theorem~\ref{XXX-thm-nearly-bal-cut}) with parameter $\phi_0$ on $G\{U\}$. Let $C$ be the output subset.
    \begin{enumerate}
        \item If $C = \emptyset$, then the subgraph $G^\ast = G\{U\}$ quits Phase~1. \label{XXX-step-2a}
        \item If $C \neq  \emptyset$ and $\vol(C) \leq (\epsilon/12) \vol(U)$, then the subgraph $G^\ast = G\{U\}$ quits Phase~1 and enters Phase~2. \label{XXX-step-2b}
        \item Otherwise, remove the cut edges $E(C, U \setminus C)$ (\Remove{2}), and then we recurse on both sides $G\{C\}$ and $G\{U \setminus C\}$ of the cut.
    \end{enumerate}
\end{enumerate}
\end{framed}

We emphasize that we do not remove the cut edges in Step~\ref{XXX-step-2b} of Phase~1.

\begin{lemma}\label{XXX-lem-recursion-depth}
The depth of the recursion of Phase~1 is at most $d$.
\end{lemma}
\begin{proof}
 Suppose  there is still a component $U$ entering the depth $d+1$ of the recursion of Phase~1. Then according to the threshold for $\vol(C)$ specified in Step~\ref{XXX-step-2b}, we infer that $\vol(U) \leq (1 - \epsilon/12)^{d} \vol(V) < 1$ by our choice of $d$, which is impossible.
\end{proof}




\begin{framed}
\noindent {\bf Phase 2.} 

\medskip
The input graph is $G^\ast = G\{U\}$.
Define $\tau \bydef \left((\epsilon/6) \cdot \vol(U) \right)^{1/k}$.
Define the sequence: $m_1 \bydef (\epsilon/6) \cdot \vol(U)$, and $m_i \bydef m_{i-1} / \tau$, for each $1 < i \leq k+1$.
Initialize $L \gets 1$ and $U' \gets U$.
Repeatedly do the following procedure.

\begin{itemize}
\item Do the nearly most balanced sparse cut algorithm (Theorem~\ref{XXX-thm-nearly-bal-cut}) with parameter $\phi_L$ on $G\{U'\}$. Let $C$ be the output subset. Note that $\Phi_{G\{U'\}}(C) \leq \phi_{L-1}$.
\begin{itemize}
\item If $C = \emptyset$, then the subgraph $G\{U'\}$ quits Phase~2.
\item If $C \neq \emptyset$ and $\vol(C) \leq m_L / (2 \tau)$, then update $L \gets L+1$. 
\item Otherwise, update $U' \gets U' \setminus C$, and remove all edges incident to $C$ (\Remove{3}).
\end{itemize}
\end{itemize}
\end{framed}

Intuitively, in Phase~2  we keep calling the nearly most balanced sparse cut algorithm to find a cut $C$ and remove it. If we find a cut  $C$ that has volume greater than $m_L / (2\tau)$, then we make a good progress. If $\vol(C) \leq m_L / (2 \tau)$, then we learn that the volume of the most balanced sparse cut of conductance at most $\phi_L$ is at most $2 \cdot m_L / (2\tau) = m_L / \tau  = m_{L+1}$ by Theorem~\ref{XXX-thm-nearly-bal-cut}, and so we move on to the next level by setting  $L \gets L+1$. 

The maximum possible level $L$ is $k$. Since by definition $m_k / (2\tau) = 1/2 < 1$, there is no possibility to increase $L$ to $k+1$. Once we reach $L = k$, we will repeatedly run the nearly most balanced sparse cut algorithm until we get $C = \emptyset$ and quit.

When we remove a cut $C \neq \emptyset$ in Phase~2, each $u \in C$ becomes an isolated vertex with $\deg(u)$ self loops, as all edges incident to $u$ have been removed, and so in the final decomposition $V = V_1 \cup \cdots \cup V_x$ we have $V_i = \{u\}$ for some $i$. 
We emphasize that we only do the edge removal when $\vol(C) > m_L / (2 \tau)$.
Lemma~\ref{XXX-lem-expander-trim-vol}  bounds the volume of the cuts found during Phase~2.

\begin{lemma}\label{XXX-lem-expander-trim-vol}
For each $1 \leq i \leq k$, define $C_i$ as the union of all subsets $C$ found in Phase~2 when $L \geq i$. Then either $C_i = \emptyset$ or $\vol(C_i) \leq m_i$.
\end{lemma}
\begin{proof}
We first consider the case of $i=1$. Observe that the graph $G^\ast = G\{U\}$ satisfies the property that the most balanced sparse cut of conductance at most $\phi_0$ has balance at most $2(\epsilon/12) = \epsilon/6$, since otherwise it does not meet the condition for entering Phase~2. Note that all cuts we find during Phase~2 have conductance at most $\phi_0$, and so the union of them $C_1$  is also a cut of $G^\ast$ with conductance at most $\phi_0$. This implies that  $\vol(C_1) \leq (\epsilon/6) \vol(U) = m_1$. 

The proof for the case of $2 \leq i \leq k$ is exactly the same, as the condition for increasing $L$ is to have $\vol(C) \leq m_L / (2\tau)$. Let $G' = G\{U'\}$ be the graph considered in the iteration when we increase $L = i-1$ to $L = i$. The existence of such a cut $C$ of $G'$ implies that the most balanced sparse cut of conductance at most $\phi_{i-1}$ of $G'$ has volume at most $2\vol(C) \leq m_{i-1} / \tau = m_{i}$. Similarly, 
note that all cuts we find when $L \geq i$ have conductance at most $\phi_{i-1}$, and so the union of them $C_i$  is also a cut of $G'$ with conductance at most $\phi_{i-1}$. This implies that  $\vol(C_i) \leq m_{i}$. 
\end{proof}

\paragraph{Conductance of Remaining Components.}
For each $u \in V$, there are two possible ways for $u$ to end the algorithm: 
\begin{itemize}
    \item During Phase~1 or Phase~2,  the output of the  nearly most balanced sparse cut algorithm on the component that $u$ belongs to is $C = \emptyset$. In this case, the component that $u$ belongs to  becomes a component $V_i$ in the final decomposition  $V = V_1 \cup \cdots \cup V_x$.
    If $\phi'$ is the conductance parameter used in the  nearly most balanced sparse cut algorithm, then $\Phi(G\{V_i\}) \geq \phi'$. Note that $\phi' \geq \phi_k = \phi$.
    \item During Phase~2, $u \in C$ for the output $C$ of the nearly most balanced sparse cut algorithm. In this case, $u$ itself becomes a  component $V_i = \{u\}$ in the final decomposition  $V = V_1 \cup \cdots \cup V_x$. Trivially, we have $\Phi(G\{V_i\}) \geq \phi$.
\end{itemize}
Therefore, we conclude that 
each component $V_i$ in the final decomposition $V = V_1 \cup \cdots \cup V_x$ satisfies that $\Phi(G\{V_i\}) \geq \phi$.

\paragraph{Number of Removed Edges.} There are three places in the algorithm where we remove edges. We show that, for each $1 \leq j \leq 3$, the number of edges removed due to \Remove{$j$} is at most $(\epsilon/3)|E|$, and so the total number of inter-component edges in the final decomposition $V = V_1 \cup \cdots \cup V_x$ is at most $\epsilon |E|$.
\begin{enumerate}
    \item By Lemma~\ref{XXX-lem-recursion-depth}, the depth of recursion of Phase~1 is at most $d$. For each $i = 1$ to $d$, the number of edges removed due to the low diameter decomposition algorithm during depth $i$ of the recursion is at most $\beta |E|$. By our choice of $\beta$, the number of edges removed due to \Remove{1} is at most $d \cdot \beta|E| \leq (\epsilon/3)|E|$.
    \item For each edge $e \in E(C, U \setminus C)$ removed due to the nearly most balanced sparse cut algorithm in Phase~1, we charge the cost of the edge removal to some pairs $(v,e)$ in the following way.
    If $\vol(C) < \vol(U\setminus C)$, for each $v \in C$, and for each edge $e$ incident to $v$, we charge the amount $|E(C, U \setminus C)| / \vol(C)$ to $(v,e)$; otherwise, for each  $v \in U \setminus  C$, and for each edge $e$ incident to $v$, we charge the amount $|E(C, U \setminus C)| / \vol(U \setminus C)$ to  $(v,e)$.
    Note that each pair $(v,e)$ is being charged for at most $\log |E|$ times throughout the algorithm, and the amount per charging is at most $h(\phi_0)$.
    Therefore, the number of edges removed due to \Remove{2} is at most $(\log |E|) \cdot h(\phi_0) \cdot 2|E| \leq (\epsilon/3)|E|$
    by our choice of $\phi_0$.

\item By Lemma~\ref{XXX-lem-expander-trim-vol}, the summation of $\vol(C)$ over all cuts $C$ in $G^\ast = G\{U\}$ that are found and removed during Phase~2  due to \Remove{3} is at most $m_1 = (\epsilon/6) \vol(U) \le (\epsilon/3) |E|$. 
\end{enumerate}

\paragraph{Round Complexity.} During Phase~1, each vertex participates in at most $d = O( (1/\epsilon) \log n)$ times the nearly most balanced sparse cut algorithm and the low diameter decomposition algorithm. By our choice of parameters $\beta = O(\epsilon^2 / \log n)$ and $\phi_0 = O(\epsilon^2 / \log^7 n)$, the round complexity of both algorithms are $O(\poly(1/\epsilon, \log n))$, as we note that whenever we run the  nearly most balanced sparse cut algorithm, the diameter of each connected component is at most $O\left( \frac{\log^2 n}{\beta^2}\right) = O\left( \frac{\log^4 n }{\epsilon^4}\right)$.

For Phase~2, Lemma~\ref{XXX-lem-expander-trim-vol} guarantees that for each $1 \leq i \leq k$ the algorithm can stay $L = i$ for at most $2 \tau$ iterations. If we neither increase $L$ nor quit Phase~2 for $2 \tau$ iterations, then we have $\vol(C_L) > m_L$, which is impossible. 
Therefore, the round complexity for Phase~2 can be upper bounded by
\[2\tau \sum_{i=1}^k O(\poly(1/\phi_i, \log n)) \leq O\left(n^{2/k} \cdot \poly(1/\phi, \log n)\right).\]

During Phase~2, it is possible that the graph $G\{U'\}$ be disconnected or has a large diameter, but we are fine since we can use all edges in $G^\ast$ for communication during a sparse cut computation, and  the diameter of $G^\ast$ is at most $O\left( \frac{\log^2 n}{\beta^2}\right) = O\left( \frac{\log^4 n }{\epsilon^4}\right)$.

\section{Triangle Enumeration} \label{XXX-triangle-enum}

We show how to derive Theorem~\ref{XXX-thm-triangle-enum} by combining Theorem~\ref{XXX-thm-expander-decomposition} with other known results in~\cite{ChangPZ19,GhaffariKS17}.

\restatetriangle*

Chang, Pettie, and Zhang~\cite{ChangPZ19} showed that given an  $(\epsilon, \phi)$-\decomposition\ $V = V_1 \cup \ldots \cup V_x$ with $\epsilon \leq 1/6$, there is an algorithm $\mathcal{A}$ that finds an edge subset $E^\ast \subseteq E$ with  $|E^\ast| \leq |E|/2$ such that each triangle in $G$ is detected by some vertex during the execution of $\mathcal{A}$, except the triangles whose three edges are all  within  $|E^\ast|$. The algorithm $\mathcal{A}$ has to solve  $\tilde{O}(n^{1/3})$ times the following routing problem in each $G[V_i]$. 
Given a set of routing requests where each vertex $v$ is a source or a destination for at most $O(\deg(v))$ messages of $O(\log n)$ bits, the goal is to deliver all messages to their destinations. Ghaffari, Khun, and Su~\cite{GhaffariKS17} showed that this routing problem can be solved in $2^{O(\sqrt{\log n \log \log n})} \cdot O(\mix)$ rounds. This  was later improved to $2^{O(\sqrt{\log n })}  \cdot O(\mix)$ by Ghaffari and Li~\cite{GhaffariL2018}.

 Applying our distributed expander decomposition algorithm (Theorem~\ref{XXX-thm-expander-decomposition}), we can find  an $(\epsilon, \phi)$-\decomposition\ with $\epsilon \leq 1/6$ and  $\phi =  1/ O(\poly \log n)$  in $o(n^{1/3})$ rounds by selecting $k$ to be a sufficiently large constant.
The mixing time $\mix$ of each component $G[V_i]$ is at most $O\left(\frac{\log n}{\phi^2}\right) = O(\poly \log n)$.
 Then we apply the above algorithm $\mathcal{A}$, and it takes $2^{O(\sqrt{\log n })}  \cdot O(\mix) = 2^{O(\sqrt{\log n })}$ rounds with the routing algorithm of Ghaffari and Li~\cite{GhaffariL2018}. After that, we recurse on the edge set $E^\ast$, and we are done enumerating all triangles after $O(\log n)$ iterations. This concludes the $O(n^{1/3}) \cdot 2^{O(\sqrt{\log n})}$-round algorithm for triangle enumeration.

To improve the complexity to $\tilde{O}(n^{1/3})$, we make the  observation that the routing algorithm of~\cite{GhaffariKS17} can be seen as a distributed data structure with the following properties.
\begin{description}
\item[Parameters:]The parameter $k$ is a positive integer that specifies the depth of the hierarchical structure in the routing algorithm. Given $k$,  define $\beta$ as the number such that  $k = \log_{\beta} m$, where  $m$ is the total number of edges. 
\item[Pre-processing Time:] The algorithm for building the data structure consists of two parts. The round complexity for building the hierarchical structure is  $O(k \beta)  (\log n)^{O(k)} \cdot O(\mix)$~\cite[Lemma 3.2]{GhaffariKS17}. The round complexity for adding the portals is  $O(k  \beta^2  \log n) \cdot O(\mix)$~\cite[Lemma 3.3]{GhaffariKS17}
\item[Query Time:] After building the data structure, each routing task can be solved in $(\log n)^{O(k)} \cdot O(\mix)$ rounds~\cite[Lemma 3.4]{GhaffariKS17}.
\end{description}

The parameter $k$ can be chosen as any positive integer. In~\cite{GhaffariKS17} they used $k = \Theta(\sqrt{\log n / \log \log n} )$ to balance the pre-processing time and the query time to show that the routing task can be solved in $2^{O(\sqrt{\log n \log \log n})}  \cdot O(\mix)$ rounds. This round complexity was later improved to $2^{O(\sqrt{\log n })}  \cdot O(\mix)$ in~\cite{GhaffariL2018}.
We however note that the algorithm of~\cite{GhaffariL2018} does not admit a trade-off as above. The main reason is their special treatment of the base layer $G_0$ of the hierarchical structure. In~\cite{GhaffariL2018}, $G_0$ is a random graph with degree $2^{O(\sqrt{\log n})}$, and  simulating one round in $G_0$ already costs $2^{O(\sqrt{\log n})} \cdot \mix$ rounds in the original graph $G$.


In the triangle enumeration algorithm $\mathcal{A}$, we need to query this distributed data structure for $\tilde{O}(n^{1/3})$ times. It is possible to set $k$ to be a large enough constant so that the pre-processing time costs only $o(n^{1/3})$ rounds, while the query time is still $O(\poly \log n)$.  This implies that the triangle enumeration problem can be solved in  $\tilde{O}(n^{1/3})$  rounds.

\section{Open Problems}

In this paper, we designed a new expander decomposition algorithm that get rids of the low-arboricity part needed in~\cite{ChangPZ19}, and this implies that triangle enumeration can be solved in $\tilde{O}(n^{1/3})$ rounds, which is optimal up to a polylogarithmic factor.

Many interesting problems are left open. In particular, the current exponent of the polylogarithmic gap between the lower and the upper bounds is enormous.
The huge exponent is caused by the inefficient trade-off between the parameters in the (i) hierarchical routing structure  and the (ii) expander decomposition algorithm.
Improving the current state of the art of (i) and (ii) will lead to an improved upper bound for triangle enumeration, as well as several other problems~\cite{Daga2019distributed,GhaffariKS17,GhaffariL2018}.

We note that the lower bound graph underlying the $\Omega(n^{1/3} / \log n)$ lower bound~\cite{IzumiL17,PanduranganRS18} for triangle enumeration is the \Erdos-\Renyi\ random graph $\mathcal{G}(n,p)$ with $p=1/2$. Hence it does not rule out the possibility of an $n^{(1/3) - \Omega(1)}$-round $\CONGEST$ algorithm for the enumeration problem on sparse graphs (i.e. $m = o(n^2)$) or the detection problem. It remains an open problem to find the asymptotically optimal round complexity of these problems in $\CONGEST$. For the case of $\CLIQUE$, efficient algorithms for these problems are already known: triangle detection can be solved in $\tilde{O}(n^{1-(2/\omega) + o(1)}) = o(n^{0.158})$~time~\cite{Censor2016}, triangle enumeration on $m$-edge graphs can be solved in $\max\{O(m/n^{5/3}), O(1)\}$~time~\cite{censorhillel_et_al:LIPIcs:2018:10064,PanduranganRS18}.

We  would also like to further investigate the power of the distributed expander decomposition. Can this tool be applied to other distributed problems than triangle detection and enumeration? It has been known that this technique can be applied to give a sublinear-time distributed algorithm for {\em exact} minimum cut~\cite{Daga2019distributed}. We expect to see more applications of distributed expander decomposition in the future.
\section*{Acknowledgment}
We thank Seth Pettie for very useful discussion.
\bibliographystyle{abbrv}
\bibliography{references}

\newpage
\appendix
\section*{\huge Appendix}
\vspace{1cm}
\section{Nearly Most Balanced Sparse Cut} \label{XXX-sect-balanced-cut}

The goal of this section is to prove the following theorem.

\restatenearlybalcut*
\begin{proof}
This theorem follows from a re-parameterization of Lemma~\ref{XXX-lem-partition-nibble} and  Lemma~\ref{XXX-lem-partition-nibble-impl}.
\end{proof}

We will prove this theorem by adapting the nearly most balanced sparse cut algorithm of Spielman and Teng~\cite{spielman2004nearly}\footnote{There are many versions of the paper~\cite{spielman2004nearly}; we refer to~\url{https://arxiv.org/abs/cs/0310051v9}.} to $\CONGEST$ in a white-box manner. Before presenting the proof, we highlight the major differences between this work and the sequential algorithm of~\cite{spielman2004nearly}.
The procedure \nibble\ itself is not suitable for a distributed implementation, so we follow the idea of~\cite{ChangPZ19} to consider an approximate version of \nibble\  (Section~\ref{subsect-apx-nibble}) and use the distributed implementation described in~\cite{ChangPZ19} (Section~\ref{XXX-sect-partition-impl}).
The nearly most balanced sparse cut algorithm of Spielman and Teng~\cite{spielman2004nearly}
involves doing  $\tilde{O}(|E|)$  iterations of \nibble\ with a random starting vertex  on the {\em remaining subgraph}. We will show that this sequential process can be partially parallelized at the cost of worsening the conductance guarantee by a polylogarithmic factor (Section~\ref{subsect-par-nibble}).

\paragraph{Terminology.} Given a parameter $\phi \in (0,1)$, We define the following functions as in~\cite{spielman2004nearly}.
\begin{align*}
    \ell &\bydef \ceil{\log |E|},\\
    \tlast &\bydef 49 \ln ( |E| e^2) / \phi^2,\\
    \fff(\phi) &\bydef \frac{\phi^3 }{14^4 \ln^2 ( |E| e^4) },\\
    \gamma &\bydef \frac{5\phi}{ 7 \cdot 7 \cdot 8 \cdot \ln( |E| e^4 ) },\\
    \epsilon_b &\bydef  \frac{\phi}{ 7 \cdot 8 \cdot \ln( |E| e^4 ) \tlast 2^b}.
\end{align*}

Let $A$ be the adjacency matrix of the graph $G=(V,E)$.
We assume a 1-1 correspondence between $V$ and $\{1, \ldots, n\}$.
In a \textit{lazy random walk}, the walk stays at the current vertex with probability $1/2$ and otherwise moves to a random neighbor of the current vertex.
The matrix realizing this walk can be expressed as
 $M = (A D^{-1} + I) / 2$,
  where $D$ is the diagonal matrix with
  $(\deg(1), \ldots , \deg(n))$ on the diagonal.
  
Let $p_t^v$ be the probability distribution of the lazy random walk that begins at $v$ and walks for $t$ steps.
In the limit, as $t\rightarrow \infty$, 
$p_t(x)$ approaches $\deg(x)/ (2 |E|)$, so it is natural to measure 
$p_t(x)$ \emph{relative} to this baseline.
\[
\rho_t(x)=p_t(x)/\deg(x),
\]

Let $p:  V \mapsto [0,1]$ be any function. 
The truncation operation $[p]_\epsilon$ 
rounds $p(x)$ to zero if it falls below a threshold that depends on $x$.
\[
  [p]_{\epsilon} (x)
=
\begin{cases}
  p (x) & \text{if $p (x) \geq 2\epsilon \deg(x)$,}\\
  0 & \text{otherwise}.
\end{cases}
\]

As in~\cite{spielman2004nearly}, for any vertex set $S$, we define the vector  $\chi_S$ by $\chi_S(u) = 1$ if $u \in S$ and $\chi_S(u) = 0$ if $u \notin S$, and we define the vector  $\psi_S$ by $\psi_S(u) = \deg(u) / \vol(S)$ if $u \in S$ and $\psi_S(u) = 0$ if $u \notin S$. In particular, $\chi_v$ is a probability distribution on $V$ that has all its probability mass on the vertex $v$, and $\psi_V$ is the degree distribution of $V$. That is, $\Prob_{x \sim \psi_V}[x = v] = \deg(v)/\vol(V)$.

\subsection{Nibble}
We first review the \nibble\ algorithm of~\cite{spielman2004nearly}, which computes the following  sequence of vectors with truncation parameter $\epsilon_b$.
\begin{align*}
\ppp_t &=\begin{cases}
\chi_v & \text{if $t = 0$,}\\
[M \ppp_{t-1}]_{\epsilon_b} & \text{otherwise.}
\end{cases}
\end{align*}
We define $\rhoo_t(v) = \ppp_t(v) / \deg(v)$ as the normalized probability mass at $v$ at time $t$.
Due to truncation,  for all $u \in V$ and $t \geq 0$, we have $p_t(u) \geq \ppp_t(u)$ and $\rho_t(u) \geq \rhoo_t(u)$.

We define $\tpi$ as a permutation of $V$ such that $\rhoo_t(\tpi(1)) \geq  \rhoo_t(\tpi(2)) \geq \cdots \rhoo_t(\tpi(|V|))$. That is, we order the vertices by their  $p(v)/\deg(v)$-value, breaking ties arbitrarily (e.g., by comparing IDs).
We write $\tpi(i \ldots j)$ to denote the set of vertices $\tpi(x)$  with $i \leq x \leq j$. For example, $\tpi(1 \ldots j)$ is the set of  the top $j$ vertices with the highest $\rhoo(v)$-value.

\begin{framed}
\noindent {\bf Algorithm} \nibble($G,v,\phi,b$)

\medskip

For $t  = 1$ to $\tlast$, if there exists an index $1 \leq j \leq |V|$ meeting the following  conditions 
\begin{enumerate}
\addtolength{\itemindent}{1cm}
    \item[(C.1)] $\Phi(\tpi(1 \ldots j)) \leq \phi$.
    \item[(C.2)] $\rhoo_t(\tpi(j)) \geq \gamma / \vol(\tpi(1 \ldots j))$.
    \item[(C.3)] $(5/6) \vol(V) \geq \vol(\tpi(1 \ldots j)) \geq (5/7) 2^{b-1}$. 
\end{enumerate}
then return $C = \tpi(1 \ldots j)$ and quit.
Otherwise return $C = \emptyset$.
\end{framed}

Note that the definition of \nibble($G,v,\phi,b$) is exactly the same as the one presented in~\cite{spielman2004nearly}.

\begin{definition}
 Define $Z_{u,\phi,b}$ as the subset of $V$ such that if we start the lazy random walk from $v \in Z_{u,\phi,b}$, then $\rho_t(u) \geq \epsilon_b$ for at least one of  $t \in [0, \tlast]$. 
 For any edge $e = \{u_1, u_2\}$, define $Z_{e, \phi, b} = Z_{u_1,\phi,b} \cup Z_{u_2,\phi,b}$.
 \end{definition}
 
 Intuitively, if $v \notin Z_{e, \phi, b}$, then $e$ does not  participate in \nibble($G,v,\phi,b$) and  both endpoints of $e$ are not in the output $C$   of \nibble($G,v,\phi,b$).
In particular, 
  $v \in Z_{e, \phi, b}$ is  a necessary  condition for $e \in E(C)$,
The following auxiliary lemma establishes upper bounds on $\vol(Z_{u,\phi,b})$ and $\vol(Z_{e,\phi,b})$. This lemma will be applied to bound the amount of congestion when we execute multiple \nibble\ in parallel. Intuitively, if $\vol(Z_{e,\phi,b})$ is small, then we can afford to run many instances \nibble($G,v,\phi,b$) in parallel for random starting vertices $v$ sampled from the degree distribution $\psi_V$.

\begin{lemma}\label{XXX-lem-nibble-aux}
The following formulas hold for each vertex $u$ and each edge $e$.
\begin{align*}
    &\vol(Z_{u,\phi,b}) \leq (\tlast+1) / (2\epsilon_b)\\
    &\vol(Z_{e,\phi,b}) \leq (\tlast+1) / \epsilon_b
\end{align*}
In particular, these two quantities are both upper bounded by $O(\phi^{-5} 2^b \log^3{|E|})$.
\end{lemma}
\begin{proof}
In this proof we use superscript to indicate the starting vertex of the lazy random walk.
We write $Z_{u,\phi,b,t} = \{ v \in V \ | \ \rho_t^v(u) \geq 2\epsilon_b \}$. Then $\vol(Z_{u,\phi,b}) \leq \sum_{t = 0}^{\tlast} \vol(Z_{u,\phi,b,t})$. Thus, to prove the lemma, if suffices to show that $\vol(Z_{u,\phi,b,t}) \leq 1 / (2\epsilon_b)$.
This inequality follows from the fact that $\rho_t^v(u) = \rho_t^u(v)$, as follows.
\begin{align*}
1 &= \sum_{v \in V} p_t^u(v) \\
&\geq  \sum_{v \in V \ | \ \rho_t^u(v) \geq 2\epsilon_b} p_t^u(v) \\
&\geq \sum_{v \in V \ | \ \rho_t^u(v) \geq 2\epsilon_b} 2\epsilon_b \cdot \deg(v) \\
&= \sum_{v \in V \ | \ \rho_t^v(u) \geq 2\epsilon_b} 2\epsilon_b \cdot \deg(v) \\
&= 2\epsilon_b \cdot  \vol(Z_{u,\phi,b,t}).
\end{align*}
The fact that $\rho_t^v(u) = \rho_t^u(v)$ as been observed in~\cite{SpielmanT13} without a proof. 
For the sake of completeness, we will show a proof of this fact. An alternate proof can be found in~\cite[Lemma 3.7]{ChangPZ19}. In the following calculation, we use the fact that $D^{-1} M D = D^{-1}(AD^{-1}+I)D/2 = (D^{-1}A + I)/2 =  M^{\top}$.
\begin{align*}
    \rho_t^v(u) &= \chi_u^\top D^{-1} M^t \chi_v\\
    &= \chi_u^\top (D^{-1} M D)^t   (D^{-1}\chi_v)\\
    &= \chi_u^\top (M^{\top})^t   (D^{-1}\chi_v)\\
    &=  (D^{-1}\chi_v)^\top M^t \chi_u \\
    &=  \chi_v^\top D^{-1}  M^t \chi_u \\
    &= \rho_t^u(v).
\end{align*}
Finally, recall that $\epsilon_b = \frac{\phi}{ 7 \cdot 8 \cdot \ln( |E| e^4 ) \tlast 2^b}$ and $\tlast = 49 \ln ( |E| e^2) / \phi^2$, and so 
\[ \vol(Z_{e,\phi,b}) \leq  2(\tlast + 1)/ (2\epsilon_b) =  O(\phi^{-5} 2^b \log^3{|E|}). \qedhere\]
\end{proof}

 Lemma~\ref{XXX-lem-nibble} lists some crucial properties of \nibble.   In subsequent discussion, for any given subset $S\subset V$, the subset $S^g \subseteq S$ and the partition $S^g = \bigcup_{b=1}^{\ell} S^g_b$  are defined according to Lemma~\ref{XXX-lem-nibble}.

\begin{lemma}[Analysis of \nibble]\label{XXX-lem-nibble}
For each $\phi \in (0,1]$, and for each subset $S\subset V$ satisfying
\begin{align*}
\vol(S) &\leq \frac{2}{3}\cdot \vol(V)
\text{ \ \ and \ \ \ } 
\Phi(S) \leq 2 \fff(\phi),
\end{align*}
there exists a subset $S^g \subseteq S$ with the following properties.
First, $\vol(S^g)\geq \vol(S)/2$.  Second, $S^g$ is partitioned into
$S^g = \bigcup_{b=1}^{\ell} S^g_b$ such that if a lazy random walk is initiated
at any $v\in S^g_b$ with truncation parameter
$\epsilon_b$, the following are true.
\begin{enumerate}
    \item\label{XXX-N1} The set $C$ returned by \nibble$(G,v,\phi,b)$ is non-empty.
    \item\label{XXX-N2} 
    Let $1/5 < \lambda$. For any $1 \leq t \leq  \tlast$ and $j$ satisfying $\rhoo_t(\tpi(j)) \geq \lambda \gamma / \vol(\tpi(1 \ldots j))$,
    we have   $\vol(\tpi(1 \ldots j) \cap S) \geq (1 - \frac{1}{5 \lambda})  \vol(\tpi(1 \ldots j))$. In particular, the  set $C$ returned by \nibble$(G,v,\phi,b)$  satisfies $\vol(C \cap S) \geq (4/7)2^{b-1}$.
\end{enumerate}
\end{lemma}
\begin{proof}
The first condition follows from~\cite[Lemma 3.1]{spielman2004nearly}. The second condition follows from the proof of~\cite[Lemma 3.14]{spielman2004nearly}. To see that the  set $C$ returned by \nibble$(G,v,\phi,b)$  satisfies $\vol(C \cap S) \geq (4/7)2^{b-1}$, observe that by (C.3),  the  set $C$ satisfies $\vol(C) \geq (5/7)2^{b-1}$. Setting $\lambda = 1$, (C.2) implies that $\vol(C \cap S) \geq (1 - \frac{1}{5})  \vol(C) \geq (4/5)(5/7)2^{b-1} = (4/7)2^{b-1}$.
\end{proof}

To put it another way, Lemma~\ref{XXX-lem-nibble}(\ref{XXX-N1}) says that there exist  $1 \leq t \leq  \tlast$ and $1 \leq j \leq |V|$ such that (C.1)--(C.3) are met;   Lemma~\ref{XXX-lem-nibble}(\ref{XXX-N2}) says that if $t$ and $j$ satisfy (C.2), then   the set $\tpi(1 \ldots j)$ has high overlap with $S$.

Intuitively, the set $S^g$ represents the ``core'' of $S$ in the sense that \nibble$(G,v,\phi,b)$  is guaranteed to return a  sparse cut $C$ if $v \in S_b^g$. Recall that (C.1) and (C.3) in the description of \nibble$(G,v,\phi,b)$ guarantees that the cut $C$ has conductance at most $\phi$ and has volume at most $(5/6)\vol(V)$. 

\subsection{Approximate Nibble}\label{subsect-apx-nibble}
The  algorithm \nibble\ is not suitable for a distributed implementation since it has to go over all possible $j$. Similar to the idea of~\cite[Algorithm~1]{ChangPZ19} we provide a slightly modified version of \nibble\ that only considers $O(\phi^{-1} \log \vol(V) )$ choices of $j$ for each $t$.
The cost of doing so is that we have to relax the  conditions  slightly.

Given a number $t$, we define the sequence $(j_x)$ as follows. We write $\jmax$ to denote the largest index with $\ppp
_t(\jmax) > 0$.
For the base case, $j_1 = 1$.  Now suppose $j_1, \ldots, j_{i-1}$ as been defined.
If we already have $j_{i-1} = \jmax$, then we are done, i.e., $j_{i-1} = \jmax$ is the last element of the sequence $(j_x)$; 
otherwise,  the next element  $j_i$ is selected as follows.
\[
j_i = \max\left\{j_{i-1} + 1,  \; \; \operatorname{arg \ max}_{1 \leq j \leq \jmax}  \left( \vol(\tpi(1 \ldots j)) \leq (1+\phi) \vol(\tpi(1 \ldots j_{i-1})) \right)\right\}.
\]

\begin{framed}
\noindent {\bf Algorithm} \distnibble($G,v,\phi,b$)

\medskip

For $t  = 1$ to $\tlast$, we go over all $O(\phi^{-1} \log \vol(V) )$ candidates $j$ in the sequence $(j_x)$.  If $j_x = 1$ or $j_x = j_{x-1} + 1$, we test whether (C.1), (C.2), and (C.3) are met. Otherwise, we test whether the following modified conditions are met.
\begin{enumerate}
\addtolength{\itemindent}{1cm}
    \item[(C.1*)] $\Phi(\tpi(1 \ldots j_x)) \leq 12 \phi$.
    \item[(C.2*)]   $\rhoo_t(\tpi(j_{x-1})) \geq \gamma / \vol(\tpi(1 \ldots j_x))$.
    \item[(C.3*)] $(11/12) \vol(V) \geq \vol(\tpi(1 \ldots j_x)) \geq (5/7) 2^{b-1}$. 
\end{enumerate}
If some $j_x$ passes the test, then return $C = \tpi(1 \ldots j_x)$ and quit.
Otherwise return $C = \emptyset$.
\end{framed}


\begin{definition}\label{def-participate-edges}
Consider \distnibble$(G,v,\phi,b)$. Define $P^\ast$ as the set of edges  $e$ such that there exist at least one endpoint $u$ of $e$ and at least one number $t \in [0, \tlast]$ with $\ppp_t(u) > 0$.
\end{definition}

 Intuitively,   $P^\ast$  is the set of edges that participate in   \distnibble($G,v,\phi,b$). This notation will be used in analyzing the complexity of our distributed implementation.



Lemma~\ref{XXX-lem-dist-nibble} shows an additional property of the output $C$ of \distnibble$(G,v,\phi,b)$  when $v$ is appropriately chosen.
Note that if   $C$  is non-empty,   it must have conductance at most $12 \phi$ and volume at most $(11/12)\vol(V)$ in view of (C.1*) and (C.3*). 

\begin{lemma}[Analysis of \distnibble]\label{XXX-lem-dist-nibble}
For each $0 < \phi \leq 1/12$, and for each subset $S\subset V$ satisfying
\begin{align*}
\vol(S) &\leq \frac{2}{3}\cdot \vol(V)
\text{ \ \ and \ \ \ } 
\Phi(S) \leq 2 \fff(\phi),
\end{align*}
the output $C$ of \distnibble$(G,v,\phi,b)$ for any $v \in S_g^b$ is non-empty and it satisfies 
\[\vol(C \cap S) \geq 2^{b-2}.\]
\end{lemma}
\begin{proof}
We pick $(t,j)$ as the indices that satisfy (C.1)--(C.3), whose existence is guaranteed by Lemma~\ref{XXX-lem-nibble}(\ref{XXX-N1}). Let $v \in S_g^b$. 
We select $x$ in such a way that $j_{x-1} \leq j \leq j_x$. We will show that $j_x$ will pass the test  in \distnibble$(G,v,\phi,b)$, and the output $C = \tpi(1 \ldots j_x)$ satisfies $\vol(C \cap S) \geq 2^{b-2}$.

For the easy special case that 
$j = j_i$ for some $i$,
the index  $j_i$ is guaranteed to pass the test in \distnibble$(G,v,\phi,b)$, and we have $\vol(C \cap S) \geq (4/7) 2^{b-1} > 2^{b-2}$ by Lemma~\ref{XXX-lem-nibble}.

Otherwise, the three indices $j_{x-1} \leq j \leq j_x$ satisfy the following relation:
\[\vol(\tpi(1 \ldots j_{x-1})) \leq \vol(\tpi(1 \ldots j)) \leq \vol(\tpi(1 \ldots j_{x})) \leq (1+\phi) \vol(\tpi(1 \ldots j_{x-1})).\]
We first show that  $j_x$   satisfies the three conditions (C.1*), (C.2*), (C.3*), and so it will pass the test in \distnibble$(G,v,\phi,b)$, and then we show that the output $C$ satisfies $\vol(C \cap S) \geq 2^{b-2}$.

\paragraph{Condition (C.1*).} 
We divide the analysis into two cases.
\begin{itemize}
\item Consider the case $\vol(\tpi(1 \ldots j_x)) \leq \vol(V)/2$. We have $|\partial(\tpi(1 \ldots j_x))| \leq |\partial(\tpi(1 \ldots j))| + \phi \vol(\tpi(1 \ldots j_x)) \leq 2\phi \vol(\tpi(1 \ldots j)) \leq 2\phi \vol(\tpi(1 \ldots j_x))$. Hence 
\[\Phi(\tpi(1 \ldots j_x)) = |\partial(\tpi(1 \ldots j_x))| / \vol(\tpi(1 \ldots j_x)) \leq 2\phi,\] and so (C.1*) is met.
In the above calculation, we use the fact that $|\partial(\tpi(1 \ldots j))| \leq \phi \vol(\tpi(1 \ldots j))$, which is due to the assumption that  $(t,j)$ satisfies (C.1).

\item Consider the case $\vol(\tpi(1 \ldots j_x)) > \vol(V)/2$.
The last inequality in the following calculation  uses the fact  that $\vol(V \setminus \tpi(1 \ldots j)) \geq (1/6)\vol(V)$, which is due to the assumption that  $(t,j)$ satisfies (C.2).
\begin{align*}
    \vol(V \setminus \tpi(1 \ldots j_x)) 
    &\geq
    \vol(V \setminus \tpi(1 \ldots j)) - \phi \vol(\tpi(1 \ldots j_x))  \\
    &\geq \vol(V \setminus \tpi(1 \ldots j)) - (1/12)\vol(\tpi(1 \ldots j_x)) & \phi \geq 1/12\\
    &\geq \vol(V \setminus\tpi(1 \ldots j)) - (1/12)\vol(V)\\
    &\geq \vol(V \setminus \tpi(1 \ldots j)) / 2. &(*)
\end{align*}

We are ready to show that $\Phi(\tpi(1 \ldots j_x))  \leq 12 \phi$.
 \begin{align*}
     \Phi(\tpi(1 \ldots j_x)) 
     & = |\partial(\tpi(1 \ldots j_x))|  / \vol(V \setminus \tpi(1 \ldots j_x)) \\
     &\leq \frac{\phi \vol(V \setminus \tpi(1 \ldots j)) + \phi \vol(\tpi(1 \ldots j))}{\vol(V \setminus \tpi(1 \ldots j_x))} \\
     &\leq \frac{6\phi \vol(V \setminus \tpi(1 \ldots j))}{\vol(V \setminus \tpi(1 \ldots j_x))} & \vol(\tpi(1 \ldots j)) \leq (5/6)\vol(V)\\
     &\leq 12 \phi & \text{ use (*)}
 \end{align*}
\end{itemize}

\paragraph{Condition (C.2*).}
\begin{align*}
    \rhoo_t(\tpi(j_{x-1})) &\geq 
    \rhoo_t(\tpi(j)) & j_{x-1} \leq j \\
    &\geq  
    \gamma / \vol(\tpi(1 \ldots j))   & \text{$(t,j)$ satisfies (C.2)} \\
    &\geq
    \gamma / \vol(\tpi(1 \ldots j_x)).  & j \leq j_{x}.
\end{align*}

\paragraph{Condition (C.3*).}
\begin{align*}
    (11/12) \vol(V) 
    &>(5/6)(1+\phi) \vol(V) & \phi \leq 1/12\\
    &\geq (1+\phi)\vol(\tpi(1 \ldots j)) & \text{$(t,j)$ satisfies (C.3)} \\
    &\geq \vol(\tpi(1 \ldots j_x))\\
    &\geq \vol(\tpi(1 \ldots j)) & j \leq j_x\\
    &\geq (5/7) 2^{b-1}.  & \text{$(t,j)$ satisfies (C.3)} 
\end{align*}


\paragraph{Lower Bound of $\vol(C \cap S)$.} First of all, observe that (C.2*) implies that 
\[\rhoo_t(\tpi(j_{x-1})) \geq 
\frac{\gamma}{\vol(\tpi(1 \ldots j_{x}))}
\geq \frac{\gamma}{(1 + \phi) \vol(\tpi(1 \ldots j_{x-1}))} = 
\frac{(12/13)\gamma}{\vol(\tpi(1 \ldots j_{x-1}))}
.\]
By Lemma~\ref{XXX-lem-nibble}, we can lower bound  $\vol(C \cap S)$ as follows.
\begin{align*}
    \vol(C \cap S)
    &= \vol(\tpi(1 \ldots j_x) \cap S)\\
    & > \vol(\tpi(1 \ldots j_{x-1}) \cap S)\\
    &\geq (1 - \frac{13}{5\cdot 12})  \vol(\tpi(1 \ldots j_{x-1})) & \text{Lemma~\ref{XXX-lem-nibble}}\\
    &= (1 - \frac{13}{60})  \vol(\tpi(1 \ldots j_{x-1}))\\
    &\geq  (1 - \frac{13}{60}) \vol(\tpi(1 \ldots j_{x})) / (1 + \phi) \\
     &\geq (1 - \frac{13}{60})  (5/7)  2^{b-1} / (1 + \phi)  & \text{(C.3*)}\\
    &> 2^{b-2}. & \phi \leq 1/12 & & & \qedhere\\
\end{align*}
\end{proof}

Recall that the main goal of Section~\ref{XXX-sect-balanced-cut} is to design a distributed algorithm that finds a nearly most balanced sparse cut, so finding a cut $C$ with low conductance is not enough. This is in contrast to~\cite{ChangPZ19}, where they do not need the output cut to satisfy any balance constraint.

Following the approach of~\cite{spielman2004nearly}, to find a nearly most balanced sparse cut, we will need to take the union of the output of multiple instances of \distnibble, and the goal of the analysis is to show that the resulting vertex set has volume at least $\vol(S)/2$.  
This explains the reason why we not only need to show that $C \neq \emptyset$ but also  need to show a lower bound of $\vol(C \cap S)$ in Lemma~\ref{XXX-lem-dist-nibble}.

\subsection{Random Nibble}\label{subsect-rnd-nibble}
Note that both \nibble\ and \distnibble\ are deterministic. Next, we consider the algorithm \randnibble\ which executes \distnibble\ with a random starting vertex $v$ and a random parameter $b$. The definition of  \randnibble\ exactly the same as the corresponding one in~\cite{SpielmanT13} except that we use \distnibble\ instead of \nibble.

\begin{framed}
\noindent {\bf Algorithm} \randnibble($G,\phi$)

\medskip

Sample a starting vertex $v \sim \psi_V$ according to the degree distribution. Choose a number $b \in [1, \ell]$ with $\Prob[b = i] = 2^{-i} / (1 - 2^{-\ell})$. Execute \distnibble($G,v,\phi,b$), and return the result $C$.
\end{framed}

Recall that $P^\ast$ is the set of edges  participating in the subroutine \distnibble($G,v,\phi,b$), as defined in   Definition~\ref{def-participate-edges}.
Note that $E(C) \subseteq P^\ast$, 
 where $C$ is the output of  \randnibble($G,\phi$).

  

\begin{lemma}[Analysis of \randnibble]\label{XXX-lem-rand-nibble}
For each $0 < \phi \leq 1/12$, the following holds for the output $C$ of \randnibble$(G,\phi)$. 
\begin{enumerate}
    \item  $\Prob[e \in E(C)]  \leq \Prob[e \in P^\ast] \leq ( 56 \ell (\tlast+1) \tlast  \ln( |E| e^4) \phi^{-1} ) /\vol(V)$ for each $e \in E$.
    \item  $\Expect[\vol(C \cap S)]
    \geq \frac{\vol(S)}{8 \vol(V)}$ for each subset $S\subset V$ satisfying
\begin{align*}
\vol(S) &\leq \frac{2}{3}\cdot \vol(V)
\text{ \ \ and \ \ \ } 
\Phi(S) \leq 2\fff(\phi).
\end{align*} 
\end{enumerate}
\end{lemma}
\begin{proof}
The proof of $\Expect[\vol(C \cap S)]
    \geq \frac{\vol(S)}{8 \vol(V)}$ follows from Lemma~\ref{XXX-lem-dist-nibble}  and the proof of~\cite[Lemma 3.2]{spielman2004nearly}. An upper bound of $\Prob[e \in P^\ast]$ can be calculated using Lemma~\ref{XXX-lem-nibble-aux}. More specifically, observe that $v \in Z_{e,\phi,i}$ is a necessary condition for $e \in E(C)$ for the case $b = i$, and so we can upper bound $\Prob[e \in P^\ast]$ as follows.
\begin{align*}
    \Prob[e \in P^\ast] 
    &\leq \sum_{i=1}^{\ell} \Prob[b = i] \cdot \Prob[v \in  Z_{e,\phi,i}]\\
    &\leq \sum_{i=1}^{\ell} \Prob[b = i] \cdot \vol(Z_{e,\phi,i}) / \vol(V)\\
    &\leq \sum_{i=1}^{\ell} \frac{2^{-i}}{1 - 2^{-\ell}} \cdot  ( (\tlast+1) / \epsilon_b) /\vol(V) & \vol(Z_{e,\phi,i})  \leq (\tlast+1) / \epsilon_b\\
    &< \sum_{i=1}^{\ell} (  7 \cdot 8 \cdot (\tlast+1) \tlast  \ln( |E| e^4) \phi^{-1} ) /\vol(V) &\epsilon_b =  \frac{\phi}{ 7 \cdot 8 \cdot \ln( |E| e^4 ) \tlast 2^b}\\
    &= ( 56 \ell (\tlast+1) \tlast  \ln( |E| e^4) \phi^{-1} ) /\vol(V), 
\end{align*}
where the second inequality follows from the fact that we sample a starting vertex $v \sim \psi_V$ according to the degree distribution. 
\end{proof}


\subsection{Parallel Nibble} \label{subsect-par-nibble}
In~\cite{spielman2004nearly}, roughly speaking, it was shown that a nearly most balanced sparse cut can be found with probability $1 - p$ by \emph{sequentially} applying \nibble\ with a random starting vertex for $O(|E| \log (1/p))$ times.
After each \nibble, the output subset $C$ is removed from the underlying graph. To achieve an efficient implementation in $\CONGEST$, we need to diverge from this approach and aim at doing multiple \distnibble\ \emph{in parallel}. 

However, the na\"{i}ve approach of doing all $O(|E| \log (1/p))$ \randnibble\ in parallel does not work, since the potentially high overlap between the output subsets of different execution of \randnibble\ will destroy the required conductance constraint.

In what follows, we consider the algorithm \parallelnibble, which involves a simultaneous execution of a moderate number of \distnibble.
In the description of \parallelnibble, we say that $e$ participates in the subroutine \randnibble($G,\phi$) if $e \in P^\ast$ for the subroutine \distnibble($G,v,\phi,b$) during the execution of \randnibble($G,\phi$). For the sake of presentation, we write 
\[k \bydef \ceil{\vol(V) / ( 56 \ell (\tlast+1) \tlast  \ln( |E| e^4) \phi^{-1} ) }\]
in subsequent discussion. 

\begin{framed}
\noindent {\bf Algorithm} \parallelnibble($G, \phi$)

\medskip

For $i = 1$ to $k$, do \randnibble($G,\phi$), in parallel. Let $C_i$ be the result of the $i$th execution of \randnibble($G,\phi$). Let $U_i = \bigcup_{j=1}^i C_i$.  If there exists an edge $e$ participating in the subroutine \randnibble($G,\phi$) for more than $w \bydef 10 \ceil{\ln (\vol(V))}$ times, return $C = \emptyset$.
Otherwise select $i^\ast \in [1, k]$ to be the highest index such that $\vol(U_{i^\ast}) \leq z  \bydef (23/24)\vol(V)$. Return $C = U_{i^\ast}$.
\end{framed}

For the sake of presentation, in the statement of Lemma~\ref{XXX-lem-par-nibble} we define the function $g$ by
\[g(\phi, \vol(V)) \bydef \ceil{10w \cdot ( 56 \ell (\tlast+1) \tlast  \ln( |E| e^4) \phi^{-1} )} = O(\phi^{-5} \log^5 (|E|)).\]
In particular, we have $10w \vol(V)/k \leq g(\phi, \vol(V))$. The function $g$ will also be used in the description and the analysis of \balancedcut\ in the subsequent discussion.

\begin{lemma}[Analysis of \parallelnibble]\label{XXX-lem-par-nibble}
For each $0 < \phi \leq 1/12$ 
the following is true for the output $C$ of \parallelnibble$(G, \phi)$.
\begin{enumerate}
    \item If $C \neq \emptyset$, then  $\Phi(C) \leq 276 w \phi$. 
    \item For each subset $S\subset V$ satisfying
\begin{align*}
\vol(S) &\leq \frac{2}{3}\cdot \vol(V)
\text{ \ \ and \ \ \ } 
\Phi(S) \leq 2 \fff(\phi),
\end{align*} define the random variable $y$ as follows.
\[
y = 
\begin{cases}
\vol(S), \text{ if } \vol(C) \geq (1/24)\vol(V) \\
\vol(C \cap S), \text{ otherwise}\\
\end{cases}
\]
Then $\Expect[y] \geq \frac{k \vol(S)}{10w \vol(V)} \geq \frac{\vol(S)}{g(\phi, \vol(V))}$.
\end{enumerate}
\end{lemma}
\begin{proof}
We show that if the output subset $C$ is non-empty, then we must have $\Phi(C) \leq 276 w \phi$. 
By definition of \parallelnibble, if the output $C$ is non-empty, then each edge $e$ incident to $C$ is incident to at most $w$ of these vertex sets $C_1, \ldots, C_{i^\ast}$.
Therefore, $\vol(C) \geq (1/w)  \sum_{i=1}^{i^\ast} \vol(C_i)$.
Using the fact that the output $C_i$ of \distnibble$(G,v,\phi,b)$ has $\Phi(C_i)\leq 12\phi$, 
we upper bound $|\partial(C)|$ as follows.
\begin{align*}
    |\partial(C)| &\leq \sum_{i=1}^{i^\ast} |\partial(C_i)|\\ &\leq  \sum_{i=1}^{i^\ast} 12 \phi \vol(C_i)\\
    &\leq 12w \phi \vol(C).
\end{align*}
The threshold $z$ guarantees that $\vol(V \setminus C) \geq (1/23) \vol(C)$, and so $|\partial(C)|  \leq  12 \cdot 23 \cdot w \phi \vol(V \setminus C) = 276 w \phi \vol(V \setminus C)$. We conclude that $\Phi(C) \leq 276 w \phi$.

Next, we analyze the random variable $y$. We first observe that if $i^\ast <  k$, then $C = U_{i^\ast}$ has $\vol(C) \geq (1/24) \vol(V)$.
This is because that each $C_i$ must have $\vol(C_i) \leq (11/12) \vol(V)$ by definition of \distnibble. 
If $\vol(U_{i^\ast}) < (1/24) \vol(V)$, then $\vol(U_{i^\ast + 1}) <   (1/24) \vol(V) +  (11/12) \vol(V) <  (23/24) \vol(V)$, contradicting the choice of $i^\ast$.
Thus, for the case $i^\ast <  k$, we automatically have $y = \vol(S)$, which is the maximum possible value of $y$. In view of this, we can lower bound $\Expect[y]$ as follows.
\[
\Expect[y] \geq \Expect[\vol(U_k \cap S)] - \Prob[B] \cdot \vol(S),
\]
where $B$ is the  event that there exists an edge  participating in the subroutine \randnibble($G,\phi$) for more than $w$ times.
Note that $B$ implies $C = \emptyset$, but not vise versa.

By Lemma~\ref{XXX-lem-rand-nibble}, we know that $\Expect[\vol(C_i \cap S)] \geq \frac{\vol(S)}{8 \vol(V)}$, and this implies $\Expect[\vol(U_k \cap S)] \geq (1/w) \sum_{i=1}^{k} \Expect[\vol(C_i \cap S)] = \frac{k\vol(S)}{8w \vol(V)}$.
Therefore, to obtain the desired bound $\Expect[y] \geq \frac{k\vol(S)}{10w \vol(V)}$, it remains to show that $\Prob[B] \leq \frac{k}{40w \vol(V)}$.

If $k = 1 \leq w$, then $\Prob[B] = 0$.
In what follows, we assume $k \geq 2$, and this, together with
  the analysis of \randnibble($G,\phi$) in  Lemma~\ref{XXX-lem-rand-nibble}, implies that for each invocation of \randnibble($G,\phi$), we have
\begin{align}
\Prob[e \in P^\ast] \leq ( 56 \ell (\tlast+1) \tlast  \ln( |E| e^4) \phi^{-1} ) /\vol(V) \leq 2/k. \label{eq-1}
\end{align}
Let $e \in E$. Define $X_i = 1$ if $e$ participates in the $i$th \randnibble($G,\phi$), and define $X_i = 0$ otherwise. Set  $X = \sum_{i=1}^k X_i$.
By Formula~\ref{eq-1}, we infer that $\Expect[X] \leq 2$. By a Chernoff bound, $\Prob[X > w] \leq \exp(- 2(w-2) / 3) \ll (\vol(V))^{-2}$. By a union bound over all edges $e \in E$, we infer that $\Prob[B] < (\vol(V))^{-1} \ll \frac{k}{40w \vol(V)}$, as required.
\end{proof}

Intuitively, Lemma~\ref{XXX-lem-par-nibble} shows that we only lose a factor of $O(\log n)$ in conductance if we combine the result of $k$ parallel executions of \randnibble($G,\phi$).
We are now ready to present the algorithm for finding a nearly most balanced sparse cut, which involves executing \parallelnibble\ sequentially for $s = O(\poly(1/\phi, \log n))$ times on the remaining subgraph.

\begin{framed}
\noindent {\bf Algorithm} \balancedcut($G, \phi, p$)

\medskip
 Initialize $W_0 = V$. For $i = 1$ to $s \bydef 4 g(\phi, \vol(V)) \ceil{\log_{7/4}(1/p)}$ do the following.
\begin{enumerate}
    \item Execute \parallelnibble($G\{W_{i-1}\}, \phi$). Let the output be $C_i$. 
    \item Set $W_i = W_{i-1} \setminus C_i$. 
    \item If $\vol(W_i) \leq (47/48)\vol(V)$ or $i=s$, return $C = \bigcup_{j=1}^i C_j$ and quit.
\end{enumerate}
\end{framed}

\begin{lemma}[Analysis of \balancedcut]\label{XXX-lem-partition-nibble}
Let $C$ be the output of \balancedcut$(G, \phi)$, with  $0 <  \phi \leq 1/12$.
Then the following holds:
\begin{enumerate}
\item \label{XXX-P1} $\vol(C) \leq (47/48) \vol(V)$.
\item \label{XXX-P2} If $C \neq \emptyset$, then $\Phi(C) =  O(\phi \log |V|)$.
\item \label{XXX-P3} Furthermore, 
for each subset $S\subset V$ satisfying
\begin{align*}
\vol(S) &\leq \frac{1}{2} \cdot \vol(V)
\text{ \ \ and \ \ \ } 
\Phi(S) \leq \fff(\phi),
\end{align*}
with probability at least $1-p$, at least one of the following holds:
\begin{enumerate}
    \item \label{XXX-P3a} $\vol(C) \geq (1/48)\vol(V)$.
    \item \label{XXX-P3b} $\vol(S \cap C) \geq (1/2) \vol(S)$.
\end{enumerate}
\end{enumerate}
\end{lemma}
\begin{proof} This proof follows the framework of~\cite[Theorem 3.3]{spielman2004nearly}.

\paragraph{Proof of Condition~\ref{XXX-P1}.} 
Let $i'$ be the index such that the output subset $C$ is $\bigcup_{j=1}^{i'} C_j$. Then we have $\vol(C) \leq \vol(V \setminus W_{i'-1}) + \vol(C_{i'})$.
Since the algorithm does not terminate at the $(i'-1)$th iteration, we have $\vol(W_{i'-1}) > (47/48)\vol(V)$, and so
$\vol(V \setminus W_{i'-1}) \leq (1/48)\vol(V)$.
By the algorithm description of \parallelnibble, we have  $\vol(C_{i'}) \leq (23/24)\vol(W_{i'-1}) \leq (23/24)\vol(V)$. To summarize, we have $\vol(C) \leq (1/48)\vol(V) + (23/24)\vol(V) = (47/48)\vol(V)$.

\paragraph{Proof of Condition~\ref{XXX-P2}.} 
Note that the sets $C_1, \ldots, C_{i'}$ that constitute
 $C = \bigcup_{j=1}^{i'} C_j$ are disjoint vertex sets.
 We have $|\partial(C)| \leq \sum_{j=1}^{i'} |\partial(C_i)| \leq O(\phi \log |V|) \sum_{j=1}^{i'} \vol(C_i) =  O(\phi \log |V|) \cdot \vol(C)$, where the second inequality is due to Lemma~\ref{XXX-lem-par-nibble}.
 By Condition~\ref{XXX-P1}, we infer that $\vol(V \setminus C) \geq (1/47) \vol(C)$, and so we also have $|\partial(C)| \leq O(\phi \log |V|) \cdot \vol(V \setminus C)$. Hence $\Phi(C) = O(\phi \log |V|)$.

\paragraph{Proof of Condition~\ref{XXX-P3}.} 
 We focus on $h \bydef 4g(\phi, \vol(V))$ consecutive iterations from $i = x+1$ to  $i = x+h$, for some index $x$.
For each index $j \in [1, h]$, we write $H_j$ to denote  the event that (1) $\vol(S \cap W_{x+j-1}) \leq \vol(S)/2$ or (2) the algorithm ends prior to iteration $i = x+j$. 
We define the random variable $Y_j$ as follows. 
\[
Y_j = 
\begin{cases}
\frac{\vol(S)}{2g(\phi, \vol(V))} &\text{ if $H_j$ occurs (Case 1)} \\
\vol(W_{x+j-1} \cap S) &\text{ if } \vol(C_{x+j}) \geq (1/24)\vol(W_{x+j-1}) \text{ (Case 2)} \\
\vol(C_{x+j} \cap S) &\text{ otherwise (Case 3)}\\
\end{cases}
\]
We  claim that if $H_j$ does not occur, then the preconditions of Lemma~\ref{XXX-lem-par-nibble} are  met for the cut $S' = S \cap W_{x+j-1}$ in the graph $G' = G\{ W_{x+j-1}\}$ when we run \parallelnibble($G\{W_{x+j-1}\}, \phi$) during the $(x+j)$th iteration.
\begin{itemize}
\item We show that $\vol(S \cap W_{x+j-1}) \leq (2/3) \vol(W_{x+j-1})$, as follows. 

\begin{align*}
  \vol(S \cap W_{x+j-1}) 
  &\leq (1/2) \vol(V)  &\vol(S) \leq \vol(V)/2\\
  &< (1/2)(48/47) \vol(W_{x+j-1}) &\vol(W_{x+j-1}) > (47/48) \vol(V)\\
  &< (2/3) \vol(W_{x+j-1}).
\end{align*}

\item We show that  $\Phi_{G\{ W_{x+j-1}\}}(S) \leq 2\Phi(S) \leq 2 \fff(\phi, \vol(V)) \leq 2 \fff(\phi, \vol(W_{x+j-1}))$, where we write $\fff(\phi, r)$ to indicate the value of $\fff(\phi)$ when the underlying graph has volume $r$.

\begin{align*}
  \Phi_{G\{ W_{x+j-1}\}}(S)
  &= \frac{|E(S\cap W_{x+j-1},  W_{x+j-1} \setminus S)|}{\min\{\vol(S\cap W_{x+j-1}),\vol(W_{x+j-1} \setminus S)\}}\\
  &\leq \frac{|E(S, V \setminus S)|}{\min\{\vol(S\cap W_{x+j-1}),\vol(W_{x+j-1} \setminus S)\}}\\
  &< \frac{|E(S, V \setminus S)|}{(1/2)\min\{\vol(S),\vol(V \setminus S)\}}\\
  &= 2\Phi(S).
\end{align*}
  
The second inequality is explained as follows.
We have $\vol(S \cap W_{x+j-1}) > \vol(S)/2$ since $H_j$ does not occur, and we also have

\begin{align*}
\vol(W_{x+j-1} \setminus S) &\geq \vol(V \setminus S) - (1/48)\vol(V)\\
&\geq \vol(V \setminus S) - (1/24)\vol(V \setminus S) &\vol(V \setminus S) \geq (1/2)\vol(V)\\
&> \vol(V \setminus S)/2.
\end{align*}
\end{itemize}
Thus, we are able to use Lemma~\ref{XXX-lem-par-nibble} to infer that   that \[\Expect[Y_j \ | \ \overline{H_j}] \geq \frac{\vol(S \cap W_{x+j-1})}{g(\phi, \vol(W_{x+j-1}))} > \frac{\vol(S)}{2g(\phi, \vol(V))}.\]
In the calculation we use the two inequalities $g(\phi, \vol(V)) \geq g(\phi, \vol(W_{x+j-1}))$ and $\vol(S \cap W_{x+j-1}) > \vol(S)/2$, where the latter is due to $\overline{H_j}$.
Combining $\Expect[Y_j \ | \ \overline{H_j}] > \frac{\vol(S)}{2g(\phi, \vol(V))}$ with the trivial bound  $\Expect[Y_j \ | \ {H_j}] =  \frac{\vol(S)}{2g(\phi, \vol(V))}$, we conclude that
 \[\Expect[Y_j] \geq \frac{\vol(S)}{2g(\phi, \vol(V))}.\] 
 
We write $Y = \sum_{j=1}^h Y_j$, and we have  $\Expect[Y] \geq 2 \vol(S)$ in view of the above, as we recall $h = 4g(\phi, \vol(V))$.
We claim that we always have $Y \leq 4\vol(S)$.
We may write $Y = Y^1 + Y^2 + Y^3$, where $Y^i$ considers the part of $Y$ due to Case~$i$ in the definition of $Y_j$. It is clear that $Y^1 \leq h \cdot \frac{\vol(S)}{2g(\phi, \vol(V))} = 2 \vol(S)$. We claim that  $Y^2 \leq \vol(S)$ by observing that Case~2 can only occur at most once.  Suppose Case~2  occurs at iteration $i = x+j$. Then $\vol(C_{x+j}) \geq
(1/24)\vol(W_{x+j-1}) > (1/48)\vol(W_{x+j-1})$, which implies $\vol(W_{x+j}) \leq (47/48)\vol(W_{x+j-1}) \leq  (47/48)\vol(V)$, and so the algorithm terminates.  For Case~3, we have $Y^3 \leq \sum_{j=1}^h \vol(C_{x+j} \cap S) \leq \vol(S)$.
In view of the above, we have
\[
 (\vol(S)/2) \Prob[Y < (1/2)\vol(S)] + 4 \vol(S) (1 - \Prob[Y < (1/2)\vol(S)]) \geq  \Expect[Y] \geq  2 \vol(S),
\]
and this implies $\Prob[Y < (1/2)\vol(S)] \leq 4/7$.
We argue that $Y \geq (1/2)\vol(S)$ implies that either Condition~\ref{XXX-P3a} or Condition~\ref{XXX-P3b} holds. If Case~1 ever occurs, then the algorithm terminates before the last iteration $i=s$, and so we must have $\vol(C) \geq (1/48)\vol(V)$.
Similarly, if Case~2 ever occurs, we automatically have $\vol(C) \geq (1/48)\vol(V)$. Now assume Case~1 and Case~2 never occurs for all $j \in [1,h]$, then we have $\vol(C \cap S) \geq Y > (1/2)\vol(S)$.

We divide all $s$ iterations into $\ceil{\log_{7/4}(1/p)}$ intervals of length $h = 4g(\phi, \vol(V))$, and apply the above analysis to each of them. We conclude that with probability at least $1 - (4/7)^{\ceil{\log_{7/4}(1/p)}} \geq 1 - p$, there is at least one interval satisfying $Y \geq (1/2) \vol(S)$. In other words,   with probability at least $1 - p$, either Condition~\ref{XXX-P3a} or Condition~\ref{XXX-P3b} holds.
\end{proof}

\subsection{Distributed Implementation\label{XXX-sect-partition-impl}}

In this section we show that 
the algorithm \balancedcut($G, \phi$) can be implemented to run in $O(D \cdot \poly(\log n, 1/\phi))$ rounds in $\CONGEST$. We do not make effort in optimizing the round complexity.

\paragraph{Notations.} We often need to run our algorithms on a subgraph $G=(V,E)$ of the underlying communication network $G^\ast$, and so $|V|$ might be much smaller than the number of vertices $n$ in the actual communication network $G^\ast$. 
Nonetheless, we  still express the round complexity in terms of $n$. The parameter $n$ also indicates that the maximum allowed failure probability is $1 - 1/\poly(n)$.

In the implementation we sometimes need to broadcast certain information to the entire subgraph $G=(V,E)$ under consideration. Thus, the round complexity might depend on the parameter $D$, which is the diameter of $G$. However, in some scenario $G$ might be a subgraph of some other graph $G'$, and all edges within $G'$ can also be used for communication.
If the diameter of $G$ is much larger than the diameter of $G'$, then it is more efficient to do the broadcasting using the edges outside of $G$. In such a case, we can set $D$ as the diameter of $G'$, and our analysis still applies.

\begin{lemma}[Implementation of \distnibble]\label{XXX-lem-dist-nibble-impl}
Suppose $v$ initially knows that it is the starting vertex.
The algorithm  \distnibble$(G,v,\phi,b)$ can be implemented to run in $O\left(\frac{\log^4 n}{\phi^5}\right)$ rounds. Only the edges in $P^\ast$ participate in the computation. By the end of the algorithm, each vertex $u$ knows whether or not $u \in C$ w.h.p.
\end{lemma}
\begin{proof}
The proof is similar to the distributed implementation described in~\cite[Section 3.2]{ChangPZ19}.
First of all, the calculation of $\ppp(u)$ and $\rhoo(u)$  for each $0 \leq t \leq \tlast$ for each vertex $u \in V$ can be done in $\tlast = O\left(\frac{\log n}{\phi^2}\right)$ rounds.

Next, we have to go over all $O\left(\frac{\log n}{\phi}\right)$ choices of $x$ and all $ O\left(\frac{\log n}{\phi^2}\right)$ choices of $t$ to see if there is a pair $(t,j_x)$ meeting the required four conditions. More specifically, given $t$ and $x$, our task is the following.

\paragraph{Search for $j_x$ and $\tpi(1, \ldots, j_x)$.}
For each $x$, to compute the index $j_x$, we need to search for the index
$j^\ast = \operatorname{arg \ max}_{1 \leq j \leq \jmax}  \left( \vol(\tpi(1 \ldots j)) \leq (1+\phi) \vol(\tpi(1 \ldots j_{x-1})) \right)$ and then compute the set $\tpi(1, \ldots, j_x)$. This can be done in  $O(\tlast \log n) =  O\left(\frac{\log^2 n}{\phi^2}\right)$ rounds via a ``random binary search'' on the vertex set  $U$ containing all vertices $u$ with $\ppp(u) > 0$, as in~\cite{ChangPZ19}.
For the sake of presentation, we rank all vertices $u_1, \ldots, u_{|U|}$ by the ordering $\tpi$. Note that each $u_i$ does not know its rank $i$, and we cannot afford to compute the rank of all vertices in $U$.

We maintain two indices $L$ and $R$ that control the search space. Initially, $L \gets 1$ and $R \gets \jmax$. In each iteration, we  pick one vertex $u_i$ from $\{u_L, \ldots, u_R\}$ uniformly at random, and calculate $\vol(\tpi(1, \ldots, i))$. This can be done in $O(\tlast)$ rounds. More specifically, we build a spanning tree $T$ of the edge set $P^\ast$ rooted at $v$, and use only this tree for communication. It is clear that the subgraph induced by $P^\ast$ is connected and has diameter $O(\tlast)$.
To sample a vertex from the set $\{u_L, \ldots, u_R\}$ uniformly at random, we first do a bottom-up traversal to let each vertex $u$ in the tree compute the number of vertices in $\{u_L, \ldots, u_R\}$ that is within the subtree rooted at $u$. Using this information, we can sample one vertex from the set $\{u_L, \ldots, u_R\}$ uniformly at random by a top-down traversal.

If $\vol(\tpi(1, \ldots, i)) <  (1+\phi) \vol(\tpi(1 \ldots j_{x-1}))$, we update $L \gets i$; if $\vol(\tpi(1, \ldots, i))  =  (1+\phi) \vol(\tpi(1 \ldots j_{x-1}))$, we update $L \gets i$ and $R \gets i$; otherwise we update  $R \gets i-1$.  We are done when we reach $L = R$.

In each iteration, with probability $1/2$ the rank of the vertex we sample lies in the middle half of $[L,R]$,  and  so the size of search space 
$[L,R]$ is reduced by a factor of at least $3/4$. 
Thus, within $O(\log n)$ iterations we have $L = R$, and $S_j(q_t) = \{u_1, \ldots, u_j\}$ with  $j = L = R$.
The round complexity of this procedure is $O(\tlast \log n) =  O\left(\frac{\log^2 n}{\phi^2}\right)$.

\paragraph{Checking  (C.1)-(C.3) or  (C.1*)-(C.3*).} Given the index $j_x$ and the subset $\tpi(1, \ldots, j_x))$, it is straightforward to check whether these  conditions are met in $O(\tlast)$ rounds.


\paragraph{Round Complexity.} To summarize,  we go over all $O\left(\frac{\log n}{\phi}\right)$ choices of $x$ and all $ O\left(\frac{\log n}{\phi^2}\right)$ choices of $t$, and for each pair $(t,x)$ we have to spend $O\left(\frac{\log^2 n}{\phi^2}\right)$ rounds. Therefore, the total round complexity is $O\left(\frac{\log^4 n}{\phi^5}\right)$.
\end{proof}

\begin{lemma}[Implementation of \parallelnibble]\label{XXX-lem-parallel-nibble-impl}
The algorithm \parallelnibble$(G, \phi)$ can be implemented to run in $O\left(D \log n + \frac{\log^5 n}{\phi^5}\right)$ rounds in $\CONGEST$.
\end{lemma}
\begin{proof}
The implementation of \parallelnibble($G, \phi$) has three parts. 

\paragraph{Generation of \distnibble\ Instances.}
The first part is to generate all $k$ instances of \distnibble($G,v,\phi,b$), where the starting vertex $v \sim \psi_V$ is sampled according to the degree distribution, and $b \in [1, \ell]$ is sampled with $\Prob[b = i] = 2^{-i} / (1 - 2^{-\ell})$.

Following the idea of~\cite[Lemma 3.6]{ChangPZ19}, this task can be solved in $O(D + \log n)$ rounds, as follows. 
We build a BFS tree rooted at an arbitrary vertex $x$. For each vertex $v$, define $s(v)$ as the sum of $\deg(u)$ for each $u$ in the subtree rooted at $v$. In $O(D)$ rounds we can let each vertex $v$ learn the number $s(v)$ by a bottom-up traversal of the BFS tree. 

We let the root vertex $x$ sample the parameter $b$ for all $k$ instances of \distnibble.
Denote $K_i$ as the number of instances with $b = i$.
At the beginning, the root $x$ stores $K_i$ amount of $i$-tokens.
Let $L=\Theta(D)$ be the number of layers in the BFS tree. For $j = 1, \ldots, L$, the vertices of layer $j$ do the following. When an $i$-token arrives at $v$, 
the $i$-token disappears at $v$ with probability $\deg(v)/s(v)$
and $v$ locally generates an instance of \distnibble\ with starting vertex $v$ and parameter $b = i$;
otherwise, $v$ passes the $i$-token to a child $u$ 
with probability $\frac{s(u)}{s(v)-\deg(v)}$. Though $v$ might need to send a large amount of $i$-tokens to $u$, the only information $v$ needs to let $u$ know it the number of $i$-tokens.
Thus, for each $i$, the generation of all $K_i$ instances of \distnibble\ with a random starting vertex can be done in 
$L$ rounds. Using pipelining, we can do this for all $i$ in $O(D+\log n)$ rounds, independent of $k$.

\paragraph{Simultaneous Execution of \distnibble.} The second part is to run all $k$ instances of \distnibble\ simultaneously. If there is an edge $e$ participating in more than $w = O(\log n)$ of them, then the two endpoints of $e$ broadcast a special message $\star$ to everyone else to notify them to terminate the algorithm with $C = \emptyset$, and the broadcasting takes $D$ rounds. Otherwise, this task can be done in $O(\log n) \cdot O\left(\frac{\log^4 n}{\phi^5}\right) = O\left(\frac{\log^5 n}{\phi^5}\right)$ rounds in view of Lemma~\ref{XXX-lem-dist-nibble-impl}. Overall, the round complexity is $O\left(D + \frac{\log^5 n}{\phi^5}\right)$.

\paragraph{Selection of $i^\ast$ and $C = U_{i^\ast}$.} In the description of the algorithm \parallelnibble($G, \phi$), we assume that all \distnibble\ instances are indexed from $1$ to $k$. However, in a distributed implementation we cannot afford to do this. What we can do is to let the starting vertex $v$ of each \distnibble\ instance locally generate a random $O(\log n)$-bit identifier associated with the \distnibble\ instance. We say that an \distnibble\ instance is the $i$th instance if its identifier is ranked $i$th in the increasing order of all $k$ identifiers. 
With these identifiers, we can now use a random binary search to find $i^\ast$ and calculate $C = U_{i^\ast}$ in $O(D \log n)$ rounds w.h.p.

\paragraph{Round Complexity.} To summarize, the round complexity for the three parts are $O(D+\log n)$, $O\left(D + \frac{\log^5 n}{\phi^5}\right)$, and $O(D \log n)$. Thus, the total round complexity is $O\left(D \log n + \frac{\log^5 n}{\phi^5}\right)$.
\end{proof}

\begin{lemma}[Implementation of \balancedcut]\label{XXX-lem-partition-nibble-impl}
The algorithm \balancedcut$(G, \phi, p)$ with $p = 1/\poly(n)$ can be implemented to run in $O\left(\frac{D \log^7 n}{\phi^5}+ \frac{\log^{11} n}{\phi^{10}}\right)$ rounds in $\CONGEST$.
\end{lemma}
\begin{proof}
This lemma follows immediately from Lemma~\ref{XXX-lem-parallel-nibble-impl}, as \balancedcut($G, \phi$) consists of 
\[s = O\left(g(\phi, \vol(V)) \log(1/p)\right)  = O\left(\frac{\log^6 n}{\phi^5}\right)\] iterations of \parallelnibble\ (with $p = 1/\poly(n)$), where each of them costs $O\left(D \log n + \frac{\log^5 n}{\phi^5}\right)$ rounds (Lemma~\ref{XXX-lem-parallel-nibble-impl}), and so the total round complexity is  $O\left(\frac{D \log^7 n}{\phi^5}+ \frac{\log^{11} n}{\phi^{10}}\right)$.
\end{proof}

\section{Low Diameter Decomposition} \label{XXX-sect-low-diam-clustering}
The goal of this section is to prove the following theorem.

\restatelowdiamclustering*

\begin{proof}
This follows from Lemma~\ref{XXX-lem-low-diam-clustering-analysis} and Lemma~\ref{XXX-lem-aux-3}, with a re-parameterization $\beta' = \beta/3$.
\end{proof}

Our algorithm is based on the algorithm \mpx$(\beta)$,
described by Miller, Peng, and Xu~\cite{miller2013parallel}.
The goal of \mpx$(\beta)$
is to approximately implement the following procedure.  Each vertex $v$ samples
$\delta_v\sim \text{Exponential}(\beta)$, $\beta \in (0,1)$, and then  $v$ is assigned to the cluster of $u$ that minimizes $\dist(u,v) - \delta_u$.
The algorithm  \mpx$(\beta)$ is as follows.

\begin{framed}
\noindent{\bf Algorithm} \mpx$(\beta)$

\medskip

Every vertex $v$ picks a value $\delta_v\sim \text{Exponential}(\beta)$.
Denote the starting time  of $v$ as $\text{start}_v\gets \max\{1, \frac{2\log n}{\beta}-\lfloor \delta_v\rfloor\}$.
There are $\frac{2\log n}{\beta}$ epochs numbered 1 through $\frac{2\log n}{\beta}$.
At the beginning of epoch $t$, each as-yet unclustered vertex $v$ does the following.
\begin{itemize}
    \item If $\text{start}_v=t$, then $v$ becomes the cluster center of its own cluster.
    \item If $\text{start}_v > t$ and there exists a neighbor $u \in N(v)$ that has been clustered before epoch $t$, then  $v$ joins the cluster of $u$, breaking ties arbitrarily.
\end{itemize}
\end{framed}
The presentation of the algorithm \mpx$(\beta)$ follows the one in~\cite{haeupler2016faster}.
It is clear that the algorithm \mpx$(\beta)$  can be implemented in $\CONGEST$  in
$O\left(\frac{\log n}{\beta}\right)$ rounds, and each cluster has diameter at most $\frac{4 \log n}{\beta}$. The proof of Lemma~\ref{XXX-lem-property} can be found in~\cite[Corollary 3.7]{haeupler2016faster}.

\begin{lemma}[Analysis of \mpx$(\beta)$] \label{XXX-lem-property}
In the algorithm \mpx$(\beta)$, the probability of an edge $\{u,v\}$ having its endpoints in different clusters is at most $2\beta$.
\end{lemma}

By linearity of expectation, Lemma~\ref{XXX-lem-property} implies that the \emph{expected number} of inter-cluster edges is at most $2\beta |E|$. However, we need this bound to hold w.h.p. One way to obtain the high probability bound is to run \mpx$(\beta)$ for $O\left(\frac{\log n}{\beta}\right)$ times, and the output of one of them will have at most $2\beta |E|$ inter-cluster edges w.h.p. However, calculating the number of inter-cluster edges needs $O(D)$ rounds,  which is inefficient if the diameter $D$ is large.

Intuitively, the main barrier needed to be overcome is the high dependence among the $|E|$ events that an edge $\{u,v\}$ has its endpoints in different clusters.
Suppose $K$ is some large enough constant. We say that an edge $e = \{u,v\}$ is \emph{good} if it satisfies
\begin{align*}
\left| E\left( N^{4(\log n)/\beta + 1} (u))\right) \cup E\left( N^{4(\log n)/\beta + 1} (v))\right) \right| \leq \beta |E| / (K \log n). 
\end{align*}
If all edges are good,  then we cannot use a Chernoff bound with bounded dependence~\cite{Pemmaraju01} to show that with probability $1 - n^{-\Omega(K)}$ the number inter-clustered edges is at most $3\beta |E|$.

Our strategy is to compute a partition $V = V_D \cup V_S$ in such a way that all edges incident to $V_S$ are good, and $V_D$ already induces a low diameter clustering where the clusters are sufficiently far away from each other. 
We note that in our distributed model we do not assume that the number of edges $|E|$ is common knowledge, so we cannot simply use $\beta |E| / (K \log n)$ as a threshold when we construct $V_S$.

\begin{framed}
\noindent{\bf Algorithm} \lowdiamdecomp$(\beta)$

\medskip
\begin{enumerate}
\item Define $a \bydef \frac{5 \log n}{\beta}$ and $b \bydef \frac{K \log n}{\beta}$, where $K$ is some large constant.
\item Construct a partition $V = V_D \cup V_S$ meeting the following conditions w.h.p.
\begin{itemize}
    \item Each connected component of $V_D$ has diameter $O(ab)$. Moreover, for any two vertices $u$ and $v$ residing in different components, we have $\dist(u, v) > a$.
    \item Each vertex $v \in V_S$ satisfies  $\left|E(N^a(v))\right| \leq |E| / b$.
\end{itemize}
\item Execute  \mpx$(\beta)$ to obtain a clustering.
Output the partition $V = V_1 \cup \cdots \cup V_x$ by only cutting the inter-clustered edges $e = \{u,v\}$ such that at least one of $u$ and $v$ is in $V_S$. 
\end{enumerate}
\end{framed}

We prove that \lowdiamdecomp$(\beta)$ outputs a low diameter decomposition w.h.p. 

\begin{lemma}[Analysis of \lowdiamdecomp]\label{XXX-lem-low-diam-clustering-analysis}
Let  $\beta \in (0,1)$. The partition  $V = V_1 \cup \cdots \cup V_x$  resulting from  \lowdiamdecomp$(\beta)$ satisfies the following conditions w.h.p.
\begin{itemize}
    \item Each component $V_i$ has diameter  $O\left(\frac{\log^2 n}{\beta^2}\right)$.
    \item The number of inter-component edges $\left(|\partial(V_1)| + \cdots + |\partial(V_x)|\right)/2$  is at most $3 \beta |E|$.
\end{itemize}
\end{lemma}
\begin{proof}
Let $U = V_i$ be a component in the partition. Since each cluster resulting from the algorithm \mpx$(\beta)$ has diameter at most $\frac{4 \log n}{\beta} < a$, the set $U$ can only contain at most one connected component of $V_D$ as a subset.

Set $d_1 = \frac{4 \log n}{\beta}$ as an upper bound on the maximum diameter of a cluster, and set  $d_2 = O(ab) = O\left(\frac{\log^2 n}{\beta^2}\right)$ as an upper bound on the diameter of a connected component of $V_D$. We infer that the diameter of $U$ is at most $2(d_1 + 1) + d_2 =  O\left(\frac{\log^2 n}{\beta^2}\right)$, as required

Next, we prove that the number of inter-component edges is $3\beta |E|$ w.h.p. 
In the following probability calculation, we assume that the partition $V = V_D \cup V_S$ is fixed in the sense that we do not consider the randomness involved in constructing $V = V_D \cup V_S$.
Define $X_e$ as the indicator random variable for the event that $e$ is an inter-component edge, and define $X = \sum_{e \in E} X_e$.

If both endpoints of $e=\{u,v\}$ are in $V_D$, then $X_e = 0$ with probability $1$. Now suppose at least one endpoint $u$ of $e=\{u,v\}$ is in $V_S$.
By our choice of $a$ and $b$, we have
\[ 
\left| E\left( N^{4(\log n)/\beta + 1} (u)\right) \cup E\left( N^{4(\log n)/\beta + 1} (v)\right) \right| 
\leq
\left| E\left( N^{a} (u)\right) \right| 
\leq 
|E|/b = {\beta}|E| / (K \log n).
\]
PSince $X_{e}$ is independent of $X_{e'}$ for all $e' \notin E\left( N^{4(\log n)/\beta + 1} (u)\right) \cup E\left( N^{4(\log n)/\beta + 1} (v)\right)$, the set of random variables $\{ X_e \ | \ e \in E \}$ has bounded dependence $d = {\beta}|E| / (K \log n)$. Set $\mu = 2\beta |E| \leq \Expect[X]$, and set $\delta = 1/2$. By a Chernoff bound with bounded dependence~\cite{Pemmaraju01}, we have \[\Prob\left[\left(|\partial(V_1)| + \cdots + |\partial(V_x)|\right)/2 \geq 3\beta|E| \right] = \Prob[X \geq  (1+\delta)\mu] \leq O(d) \cdot \exp(-\Omega(\delta^2 \mu / d)) = n^{-\Omega(K)},\]
as required.
\end{proof}

\subsection{Distributed Implementation}
In this section, we give a distributed implementation of \lowdiamdecomp, where the non-trivial part is the construction of $V = V_D \cup V_S$.
We need the some auxiliary lemmas. 

\begin{lemma}\label{XXX-lem-aux-0}
Let $E^\ast$ be a subset of $E$. 
Initially each vertex $v$  knows which edges incident to $v$ are in $E^\ast$. Given  parameters $\tau \geq 1$ and $d \geq 1$, consider the following task.
\begin{itemize}
    \item If $|E(N^d(v)) \cap E^\ast| \leq \tau$, then 
    $v$ is required to learn all edges in the set $E(N^d(v)) \cap E^\ast$.
    \item If $|E(N^d(v)) \cap E^\ast|  > \tau$, then $v$ is required to learn the fact that $|E(N^d(v)) \cap E^\ast|  > \tau$.
\end{itemize}
This task can be solved deterministically in $O\left(\tau d\right)$ rounds.
\end{lemma}
\begin{proof}
The algorithm has $d-1$ phases. Each phase takes $O(\tau)$ rounds, and so the total round complexity is $O\left(\tau d\right)$.
Suppose that at the beginning of the $i$th phase, the following invariant $\mathcal{H}_i$ is met for each vertex $v$. If $|E(N^i(v)) \cap E^\ast| \leq \tau$, then $v$ knows the list of all edges in $E(N^i(v)) \cap E^\ast$. If $|E(N^i(v)) \cap E^\ast| > \tau$, then $v$ knows this fact.
Note that $\mathcal{H}_1$ holds initially.

The algorithm for the $i$th phase is as follows. For each vertex $v$ with $|E(N^i(v)) \cap E^\ast| \leq \tau$, $v$ sends the list of all edges in $E(N^i(v)) \cap E^\ast$ to all its neighbors $N(v)$. For each vertex $v$ with $|E(N^i(v)) \cap E^\ast| > \tau$, $v$ sends a special message $\star$ to all its neighbors $N(v)$. This clearly can be implemented in $O(\tau)$ rounds.

After that, for each vertex $v$, if (1) $v$ already knew that $|E(N^i(v)) \cap E^\ast| > \tau$ or (2) $v$ received  a special message $\star$ from some vertex $u \in N(v)$, then $v$ decides that $|E(N^{i+1}(v)) \cap E^\ast| > \tau$.
Otherwise, $v$ calculates the union of all edges in $E^\ast$ it learned so far, and this edge set is exactly $E(N^i(v)) \cap E^\ast$. 
Thus, the invariant $\mathcal{H}_{i+1}$ is met by the end of the $i$th phase.  

By the end of the $(d-1)$th phase, the invariant $\mathcal{H}_{d}$ is met, which implies the correctness of the algorithm.
\end{proof}

\begin{lemma}\label{XXX-lem-aux-00}
Given parameters $d \geq 1$, $z \geq 1$, and $0 < f < 1$, consider the following task where each vertex $v$ is required to output $0$ or $1$ meeting the following conditions.
\begin{itemize}
    \item If $|E(N^d(v))| \leq z$, then the output of $v$ is $1$ w.h.p.
    \item If $|E(N^d(v))|  \geq (1+f)z$, then the output of $v$ is $0$ w.h.p.
\end{itemize}
This task can be solved in $O\left(\frac{d \log n}{f^2}\right)$ rounds.
\end{lemma}
\begin{proof}
Let $K$ be some large constant.
If $K \log n \geq f^2 z$, then this task can be solved by running the algorithm of Lemma~\ref{XXX-lem-aux-0} with $\tau = (1+f)z$ and $E^\ast = E$. This takes $O(dz) = O\left(\frac{d \log n}{f^2}\right)$ rounds.

In subsequent discussion, we assume $K \log n <  f^2 z$.
Let $E^\ast$ be a subset of $E$ such that each edge $e \in E$ joins $E^\ast$ with probability $\frac{K \log n}{f^2 z}$.
We run the algorithm of Lemma~\ref{XXX-lem-aux-0} with $\tau = \frac{(1 + f/2)K \log n}{f^2}$. Each vertex $v$ outputs $1$ if it learned that  $|E(N^d(v)) \cap E^\ast| \leq  \frac{(1 + f/2)K \log n}{f^2}$; otherwise $v$ outputs $0$.
The round complexity of this algorithm is also $ O\left(\frac{d \log n}{f^2}\right)$.

The correctness of the algorithm can be deduced by a Chernoff bound. More specifically, let $X = |E(N^d(v)) \cap E^\ast|$.
For the case  $|E(N^d(v))| \leq z$, we have $\Expect[X] \leq \frac{K \log n}{f^2}$. Let $\mu = \frac{K \log n}{f^2}$ and $\delta = f/2$. By a Chernoff bound, we have 
\[\Prob\left[\left|E(N^d(v)) \cap E^\ast\right| \geq  \frac{(1 + f/2)K \log n}{f^2}\right] = \Prob\left[X \geq (1+\delta)\mu\right] \leq \exp\left(-\delta^2 \mu  / 3\right) = n^{-\Omega(K)}.\]
Therefore, if $|E(N^d(v))| \leq z$, then the output of $v$ is $1$ w.h.p. Similarly, we can use a Chernoff bound to show that if $|E(N^d(v))|  \geq (1+f)z$, then the output of $v$ is $0$ w.h.p. This proof of this part is omitted.
\end{proof}

\begin{lemma}\label{XXX-lem-aux-1}
Let $0 < f < 1$. There is an $O\left( \frac{d \log^2 n}{f^3}\right)$-round algorithm that  lets each vertex $v \in V$ compute an estimate $m_v$ such that  $|E(N^d(v))|/(1+f) \leq m_v \leq (1+f)|E(N^d(v))|$  w.h.p.
\end{lemma}
\begin{proof}
 Consider the following sequence: $s_1 = 1$, and $s_i = (1 + f)s_{i-1}$ for each $i > 1$. Let $i^\ast = O\left(\frac{\log n}{f}\right)$ be chosen as the largest index such that $m_{i^\ast} \leq  {n \choose 2}$. 
 For $i = 1$ to $i^\ast$, we run the algorithm of Lemma~\ref{XXX-lem-aux-00} with $z = s_i$.
 For each vertex $v$, it sets $m_v = s_{i'}$, where $i'$ is chosen as the largest index such that the algorithm of Lemma~\ref{XXX-lem-aux-00} with $z = s_{i'}$ outputs $1$.
 The correctness of this algorithm follows from Lemma~\ref{XXX-lem-aux-00}. The round complexity of this algorithm is $O\left(\frac{\log n}{f}\right) \cdot O\left(\frac{d \log n}{f^2}\right) = O\left( \frac{d \log^2 n}{f^3}\right)$.
\end{proof}

Now we are in a position to describe the construction of the  required partition $V = V_D \cup V_S$ that is used in \lowdiamdecomp$(\beta)$.

\paragraph{The Auxiliary Partition $V = V_D' \cup V_S'$.}
First, we apply Lemma~\ref{XXX-lem-aux-1} to obtain an auxiliary partition  $V = V_D' \cup V_S'$ in $O(ab\log^2 n) = O\left(\frac{\log^4 n}{\beta^2}\right)$ rounds satisfying the following conditions. 
\begin{itemize}
    \item Each vertex $v \in V_D'$ satisfies  $\left|E(N^a(v))\right| \geq \left|E(N^{100ab}(v))\right| / 2b$.
    \item Each vertex $v \in V_S'$ satisfies  $\left|E(N^a(v))\right| \leq \left|E(N^{100ab}(v))\right| / b$.
\end{itemize}
Next, we show how to obtain a desired decomposition $V = V_D \cup V_S$ in by modifying this auxiliary partition. In subsequent discussion, we assume that such a partition  $V = V_D' \cup V_S'$ is given and fixed. Note that each vertex $v \in V_S'$ already meets the requirement $\left|E(N^a(v))\right| \leq \left|E\right| / b$.

\paragraph{Construction of the Decomposition $V = V_D \cup V_S$.}
We build the set $V_D$ using the following procedure.
Initially, we set  
\[W_0 = \{ u \in V \ | \  \dist(u, V_D') \leq a\}.\]
For each iteration $i = 1$ to $\infty$,  the set $W_i$  is constructed as follows. 
For each connected component $S$ induced by $W_{i-1}$, it checks if there exists some other component $S'$ with  $\dist(S, S') \leq a$.
if so, then we add all vertices in $\{ u \in V \ | \  \dist(u, S) \leq a\}$ to  $W_{i}$; otherwise, we add all vertices in $S$ to $W_{i}$. Note that $W_i \supseteq W_{i-1}$.

The procedure terminates at iteration $i^\ast$ if no more vertex can be added, i.e., $W_{i^\ast} = W_{i^\ast + 1}$. We finalize $V_D = W_{i^\ast}$ and $V_S = V \setminus V_D$. Note that we must have $\dist(S, S') > a$ for any two distinct connected components $S$ and $S'$ induced by $V_D = W_{i^\ast}$.

\paragraph{Invariant  $\mathcal{H}$ for the Construction.}
Given any vertex set $S$, we define the following two parameters.
\begin{itemize}
    \item $N_S$ is the size of a maximum-size subset $S^\ast \subseteq S \cap V_D'$ such that $\dist (u,v) > 2a$ for each pair of distinct vertices $u$ and $v$ in $S^\ast$.
    \item $D_S$ is the diameter of $G[S]$.
\end{itemize}

We will prove that the following invariant is met throughout the procedure for each  connected component $S$ induced by $W_i$.

\begin{definition}[Invariant $\mathcal{H}$] Fix a decomposition $V = V_D' \cup V_S'$. A vertex set $S$
satisfying the following conditions is said to meet the invariant  $\mathcal{H}$.
\begin{enumerate}
    \item For each $u \in V_D'$, either $N^{a}(u) \subseteq S$ or $N^{a}(u) \cap S = \emptyset$. \label{XXX-H1}
    \item $D_S\leq 10a \cdot N_S - (4a+1)$. \label{XXX-H2}
    \item $N_S \leq 2b$. \label{XXX-H3}
\end{enumerate}
\end{definition}

Note that if $S$ satisfies $\mathcal{H}$, then the diameter of $G[S]$ is $O(ab)$.

\begin{lemma} \label{XXX-lem-auxx1}
Let $S$ be a connected component induced by $W_0 = \{ u \in V \ | \  \dist(u, V_D') \leq a\}$. Suppose $s \in S$,  $t \in S$, and $\dist_{G[S]}(s,t) \geq  (4a+1)k$, where $k$ is some positive integer.
Then there exist $k+1$ vertices $v_0, \ldots, v_{k}$ in $N_{G[S]}^{(4a+1)k}(s) \cap V_D'$ such that $\dist(v_i, v_j) > 2a$ for each $0 \leq i < j \leq k$.
\end{lemma}
\begin{proof}
We fix an $s$--$t$ shortest path $P$ in $G[S]$.
For each $i \in [0, k]$, select $u_i$ as the vertex in $P$ whose distance to $s$ in $G[S]$ is exactly $(4a+1)i$. Note that $u_0 = s$. Select $v_i$ as any vertex in $V_D'$ such that $u_i \in N^a(v_i)$. By our choice of $W_0$, such a vertex $v_i$ exists, and we must have $v_i \in S$. Moreover, \[(4a+1)i - a \leq \dist_{G[S]}(s, v_i) \leq (4a+1)i + a,\]
and it implies that $N^a(v_i)$ and $N^a(v_j)$ are disjoint for any $0 \leq i < j \leq k$.
\end{proof}

\begin{lemma} \label{XXX-lem-auxx2}
Let $S_1, \ldots S_k$ be $k$ disjoint vertex sets such that for each $u \in V_D'$, either $N^{a}(u) \subseteq S_i$ for some $1 \leq i \leq k$ or $N^{a}(u) \cap (S_1 \cup \cdots \cup S_k) = \emptyset$.
Define $S = \left\{ v \in V \ \middle| \ \dist\left(v, \bigcup_{i=1}^k S_i\right) \leq a \right\}$.
If $S$ is connected,
then the following is true.
\begin{enumerate}
    \item $N_S = \sum_{i=1}^k N_{S_i}$.
    \item $D_S \leq -1 + \sum_{i=1}^k (D_{S_i} + 2a + 1)$.
\end{enumerate}
\end{lemma}
\begin{proof}
Note that each set $S_i$ satisfies Condition~\ref{XXX-H1} of $\mathcal{H}$.
We first show that $N_S \geq \sum_{i=1}^k N_{S_i}$.
By definition of $N_{S_i}$, for each subset $S_i$, there exists a subset $S_{i}^\ast \subseteq S_i \cap V_D'$ such that $\dist (u,v) > 2a$ for each pair of distinct vertices $u$ and $v$ in $S_i^\ast$. We select $S^\ast = \bigcup_{i=1}^k S_i^\ast$. Then we also have  $\dist (u,v) > 2a$ for each pair of distinct vertices $u$ and $v$ in $S^\ast$, and hence $N_S \geq |S^\ast| = \sum_{i=1}^k N_{S_i}$.
Here we use the fact that for each $u \in V_D'$, either $N^{a}(u) \subseteq S_i$ for some $1 \leq i \leq k$ or $N^{a}(u) \cap (S_1 \cup \cdots \cup S_k) = \emptyset$. Suppose $\dist (u,v) \leq 2a$ for some $u$ and $v$ in $S^\ast$.
Then $N^a(u) \cup N^a(v)$ is connected, and so both $u$ and $v$ are within the same set $S_i^\ast \subseteq S_i$ for some $i$,
contradicting the choice of $S_i^\ast$.

Next, we show that $N_S \leq \sum_{i=1}^k N_{S_i}$.
Since $S_1, \ldots S_k$ all satisfy Condition~\ref{XXX-H1} of $\mathcal{H}$, the set $S \setminus \bigcup_{i=1}^k S_k$ contains only vertices in $V_S'$. Suppose $N_S > \sum_{i=1}^k N_{S_i}$. 
Then there exists a  set $S^\ast \subseteq S \cap V_D'$ of size $1 + \sum_{i=1}^k N_{S_i}$ such that $\dist (u,v) > 2a$ for each pair of distinct vertices $u$ and $v$ in $S^\ast$. By the pigeonhole principle, there is an index $i$ such that $|S^\ast \cap S_i| > N_{S_i}$, contradicting $N_{S_i} \geq S^\ast \cap S_i$. Therefore, we must have $N_S \leq \sum_{i=1}^k N_{S_i}$.

For the rest of the proof, we show that $D_S \leq -1 + \sum_{i=1}^k (D_{S_i} + 2a + 1)$.
Let $s \in S$ and $t \in S$ be chosen such that $\dist_{G[S]}(s,t) = D_S$. 
Let $P$ be an $s$--$t$ shortest path in $G[S]$.
Observe that $P$ can include at most $D_{S_i} + 2a + 1$ vertices from the set $\{ v \in V \ | \ \dist(v, S_i) \leq a \}$, since otherwise we can shortcut the path $P$ to obtain  a shorter  $s$--$t$  path in $G[S]$. Therefore, the number of vertices in $P$ is at most $\sum_{i=1}^k (D_{S_i} + 2a + 1)$, which implies that $D_S \leq -1 + \sum_{i=1}^k (D_{S_i} + 2a + 1)$.
\end{proof}

\begin{lemma}[Base case] \label{XXX-lem-basecase}
Let $S$ be a connected component induced by $W_0$. Then $S$ satisfies  $\mathcal{H}$.
\end{lemma}
\begin{proof}
Condition~\ref{XXX-H1} of $\mathcal{H}$ is met for $S$ by  definition of $W_0 = \{ u \in V \ | \  \dist(u, V_D') \leq a\}$. Condition~\ref{XXX-H2} of $\mathcal{H}$ is  met   for $S$ in view of Lemma~\ref{XXX-lem-auxx1}. More specifically,  Lemma~\ref{XXX-lem-auxx1} implies that $N_S \geq \floor{D_S/(4a+1)}+1 > D_S/(4a+1)$, and so $D_S < (4a+1) N_S \leq 10a \cdot N_S - (4a+1)$. Note that  $N_S \geq 1$.

For the rest of the proof, we consider Condition~\ref{XXX-H3} of $\mathcal{H}$.
We assume that Condition~\ref{XXX-H3} of $\mathcal{H}$ is not met for $S$ (i.e., $N_S > 2b$), and we will derive a contradiction.

For any $s \in S$, set $S' = N^{(4a+1) \cdot 2b}(s)$. Then we claim that $N_{S'} > 2b$. 
If $S \subseteq S'$, this statement is trivially true, as we already assume $N_S > 2b$. Otherwise, there exists a vertex $t \in S$ with $\dist_{G[S]}(s,t) > (4a+1) \cdot 2b$, and Lemma~\ref{XXX-lem-auxx1} guarantees that there exist $2b+1$ vertices $v_0, \ldots, v_{k}$ in \[N_{G[S]}^{(4a+1) \cdot 2b}(s) \cap V_D' \subseteq N^{(4a+1) \cdot 2b}(s)  \cap V_D'\] such that $\dist_G(v_i, v_j) > 2a$ for each $0 \leq i < j \leq 2b$, and so $N_{S'} > 2b$.

We choose $s$ to be a vertex in $S \cap V_{D}'$ minimizing $|E(N^a(v))|$. By the above claim, we have 
\[
\left|E\left(N^{100ab}(s)\right)\right| \geq 
\left|E\left(N^{(4a+1) \cdot 2b + a}(s)\right)\right| \geq 
N_{S'} \cdot \left|E(N^{a}(s))\right| \geq
(2b+1) \left|E(N^{a}(s))\right|,
\]
contradicting the definition of the set $V_D'$.
\end{proof}

\begin{lemma}[Inductive step] \label{XXX-lem-inductivestep}
If each connected component of $W_{i-1}$ satisfies $\mathcal{H}$, then each connected component of $W_{i}$ also satisfies $\mathcal{H}$.
\end{lemma}
\begin{proof}
Let $S$ be a connected component induced by $W_{i}$.
If $S$ itself is also a connected component  induced by $W_{i-1}$, then the lemma trivially holds.
Otherwise, there exist $k \geq 2$ connected components $S_1, \ldots S_k$  of  $W_{i-1}$ such that $S = \left\{ v \in V \ \middle| \ \dist\left(v, \bigcup_{i=1}^k S_i\right) \leq a\right\}$.
Note that Condition~\ref{XXX-H1} of $\mathcal{H}$ holds for $S$ trivially.

By Lemma~\ref{XXX-lem-auxx2}, we infer that Condition~\ref{XXX-H2} of $\mathcal{H}$  holds for $S$.
More specifically, 
\begin{align*}
D_S &\leq -1 + \sum_{i=1}^k (D_{S_i} + 2a + 1)\\
&\leq -1 + \sum_{i=1}^k ((10a\cdot N_{S_i} - (4a+1) ) + 2a + 1)\\
&\leq -1 - k\cdot 2a + \sum_{i=1}^k 10a\cdot N_{S_i}\\
&\leq -(4a+1) + \sum_{i=1}^k 10a\cdot N_{S_i}\\
& = 10a\cdot N_S - (4a+1).
\end{align*}

For the rest of the proof, we consider Condition~\ref{XXX-H3} of $\mathcal{H}$. Note that Lemma~\ref{XXX-lem-auxx2} already shows that $N_S = \sum_{i=1}^k N_{S_i}$. Our plan is to show that when $N_S > 2b$, we obtain a contradiction to the definition of $V_D'$.

Now suppose $N_S > 2b$. We assume that the sets $S_1, \ldots, S_k$ are ordered in such a way that for each $1 \leq j \leq k$, the set $S' = \left\{ v \in V \ \middle| \ \dist\left(v, \bigcup_{i=1}^j S_i\right) \leq a \right\}$ is connected. Since $N_{S_i} \leq 2b$ for each $1 \leq i \leq k$, there must be an index $1 < j \leq k$ such that $2b < {N_{S'}} = \sum_{i=1}^j N_{S_i} \leq 4b$. We fix $j$ to be any such index.
By Lemma~\ref{XXX-lem-auxx2}, we infer that $D_{S'} \leq 10a\cdot N_{S'}  - (4a+1) < 40ab$ along the line of the above calculation.
We choose $s$ to be a vertex in $S' \cap V_{D}'$ minimizing $|E(N^a(s))|$. Since $N_{S'} > 2b$ and  $D_{S'} < 40ab$,  we have 
\[
\left|E\left(N^{100ab}(s)\right)\right| \geq 
\left|E\left(N^{40ab}(s)\right)\right| \geq 
N_{S'} \cdot \left|E(N^{a}(s))\right| \geq
(2b+1) \left|E(N^{a}(s))\right|,
\]
contradicting the definition of the set $V_D'$.
\end{proof}

 Lemma~\ref{XXX-lem-basecase} and  Lemma~\ref{XXX-lem-inductivestep} together show that  the  invariant $\mathcal{H}$ is satisfied for each connected component $S$ of $W_i$   for all $i$.  This immediately implies that the output set $V_D$ meets all the required properties.
   That is, each connected component of $V_D$ has diameter $O(ab)$; also, for any two vertices $u$ and $v$ residing in different components, we have $\dist(u, v) > a$.
   The set $V_S$ also satisfies the required property 
 $\left|E(N^a(v))\right| \leq |E| / b$, since $V_S \subseteq V_S'$, and $v \in V_S'$ implies that $\left|E(N^a(v))\right| \leq \left|E(N^{100ab}(v))\right| / b \leq |E| / b$.

\paragraph{Round Complexity.}
Remember that the decomposition   $V = V_D' \cup V_S'$ can be constructed in $O(ab\log^2 n)$ rounds. 
We next show that the construction of the  decomposition $V = V_D \cup V_S$ from a given decomposition $V = V_D' \cup V_S'$ can be done in $O(ab^2)$ rounds.
The  invariant $\mathcal{H}$ implies that the procedure of constructing $V = V_D \cup V_S$ must terminate by the end of the $(2b-1)$th iteration. Therefore, we only need to show that each iteration can be implemented in $O(ab)$ rounds.

An $O(ab)$-round implementation for one iteration is as follows. First of all, each connected component $S$ of $W_{i-1}$ in $O(ab)$ rounds generates a unique identifier that is agreed by all its members. Then, using the identifiers, in $O(a)$ rounds, for each connected component $S$ of $W_{i-1}$, we can let each $v \in S$ learn whether there exists some other vertex $u \in W_{i-1} \cap {N^a(v)} \setminus S$.
If so, then in $O(ab)$ rounds we can let all vertices within distance $a$ to $S$ be notified that they are included in  $W_{i}$.

\begin{lemma}[Implementation of \lowdiamdecomp]\label{XXX-lem-aux-3}
Algorithm \lowdiamdecomp$(\beta)$ can be implemented to run in 
 $O(ab^2 + ab \log^2 n)=  O\left(\frac{\log^3 n}{\beta^3} + \frac{\log^4 n}{\beta^2}\right)$ rounds.
\end{lemma}
\begin{proof}
It follows from the above discussion.
\end{proof}

\end{document}